%% file: edgematching.tex
\documentclass[11pt]{article}
\usepackage{amsmath,amsthm,fullpage}
\usepackage{graphicx}
\usepackage{subcaption}
\usepackage{tikz}
\usepackage{authblk}
\PassOptionsToPackage{hyphens}{url} 
\usepackage[breaklinks]{hyperref}
\usepackage[hyphenbreaks]{breakurl} 
\usepackage{afterpage}
\usetikzlibrary{graphs,backgrounds,patterns,decorations.pathreplacing}

\DeclareFontShape{OT1}{cmr}{bx}{sc}{<-> cmbcsc10}{}

\newtheorem{theorem}{Theorem}
\newtheorem{lemma}{Lemma}
\newtheorem{definition}{Definition}

\newtheorem{observation}{Observation}
\newtheorem{corollary}{Corollary}

\def\MPC{\mathrm{MPC}}
\def\SAT{\mathrm{3SAT}}
\def\EMP{\mathrm{EMP}}

\def\Gap#1{\textsc{Gap} #1}
\definecolor{xxxcolor}{rgb}{0.8,0,0}

\let\realbfseries=\bfseries
\def\bfseries{\realbfseries\boldmath}

\title{Even $1 \times n$ Edge-Matching and Jigsaw Puzzles are Really Hard}
\author[1]{Jeffrey Bosboom}
\author[1]{Erik D. Demaine}
\author[1]{Martin L. Demaine}
\author[1]{Adam Hesterberg}
\author[2]{Pasin~Manurangsi\thanks{Part of this work was completed while the author was at MIT and Dropbox, Inc.}}
\author[1]{Anak Yodpinyanee\thanks{Research supported by NSF grant CCF-1420692.}}
\affil[1]{Massachusetts Institute of Technology,
		  Cambridge, MA 02139, USA \authorcr
          \url{{jbosboom,edemaine,mdemaine,achester,anak}@mit.edu}}
\affil[2]{University of California,
		  Berkeley, CA 94720, USA \authorcr
          \url{pasin@berkeley.edu}}
\date{}

\begin{document}
\maketitle

\begin{abstract}
  We prove the computational intractability of rotating and placing $n$ square
  tiles into a $1 \times n$ array such that adjacent tiles are
  compatible---either equal edge colors, as in edge-matching puzzles,
  or matching tab/pocket shapes, as in jigsaw puzzles.
  Beyond basic NP-hardness, we prove that it is
  NP-hard even to approximately maximize the number of placed tiles (allowing
  blanks), while satisfying the compatibility constraint between nonblank
  tiles, within a factor of 0.9999999851.  (On the other hand, there is
  an easy $1 \over 2$-approximation.)  This is the first (correct) proof of
  inapproximability for edge-matching and jigsaw puzzles.
  Along the way, we prove NP-hardness of distinguishing, for a directed graph
  on $n$ nodes, between having a Hamiltonian path (length $n-1$) and having
  at most $0.999999284(n-1)$ edges that form a vertex-disjoint union of paths.
  We use this gap hardness and gap-preserving reductions to establish similar
  gap hardness for $1 \times n$ jigsaw and edge-matching puzzles.
\end{abstract}

\section{Introduction}

Jigsaw puzzles \cite{Williams-2004} and edge-matching puzzles
\cite{Haubrich-1995} are two ancient types of puzzle, going back to the
1760s and 1890s, respectively.  Jigsaw puzzles involve fitting together a given
set of pieces (usually via translation and rotation) into a desired shape
(usually a rectangle), often revealing a known image or pattern.
The pieces are typically squares with a pocket cut out of or a tab attached
to each side, except for boundary pieces which have one flat side
and corner pieces which have two flat sides.
Most jigsaw puzzles have unique tab/pocket pairs that fit together, but we
consider the generalization to ``ambiguous mates'' where multiple tabs and
pockets have the same shape and are thus compatible.

Edge-matching puzzles are similar to jigsaw puzzles: they too feature
square tiles, but instead of pockets or tabs, each edge has a color or pattern.
In \emph{signed} edge-matching puzzles, the edge labels come in complementary
pairs (e.g., the head and tail halves of a colored lizard), and adjacent tiles
must have complementary edge labels on their shared edge (e.g., forming an
entire lizard of one color).  This puzzle type is essentially identical to
jigsaw puzzles, where complementary pairs of edge labels act as identically
shaped tab/pocket pairs.
In \emph{unsigned} edge-matching puzzles, edge labels are arbitrary, and
the requirement is that adjacent tiles must have identical edge labels.
In both cases, the goal is to place (via translation and rotation) the tiles
into a target shape, typically a rectangle.

A recent popular (unsigned) edge-matching puzzle is \emph{Eternity II}
\cite{eternity2-wiki}, which featured a US\$2,000,000 prize for the first
solver (before 2011).  The puzzle remains unsolved (except presumably by its
creator, Christopher Monckton).  The best partial solution to date
\cite{eternity2-Sydsvenskan} either places 247 out of the 256 pieces without
error, or places all 256 pieces while correctly matching 467 out of 480 edges.

\paragraph{Previous work.}
The first study of jigsaw and edge-matching puzzles from a computational
complexity perspective proved NP-hardness \cite{Jigsaw_GC}.
Four years later, unsigned edge-matching puzzles were
proved NP-hard even for a target shape of a $1 \times n$ rectangle
\cite{EFW11}.
There is a simple reduction from unsigned edge-matching puzzles to
signed edge-matching/jigsaw puzzles \cite{Jigsaw_GC}, which expands the
puzzle by a factor of two in each dimension, thereby establishing NP-hardness
of $2 \times n$ jigsaw puzzles.
Unsigned $2 \times n$ edge-matching puzzles were claimed to be APX-hard
(implying they have no PTAS) \cite{Antoniadis-Lingas-2010}, but the proof
is incorrect.%
\footnote{Personal communication with Antonios Antoniadis, October 2014.
  In particular, Lemma 3's proof is incomplete.}

\paragraph{Our results.}
We prove that $1 \times n$ jigsaw puzzles and $1 \times n$ edge-matching
puzzles are both NP-hard, even to approximate within a factor of
0.9999999851 ($ > {67038719 \over 67038720}$).
This is the first correct inapproximability result for either problem.
Even NP-hardness is new for $1 \times n$ signed edge-matching/jigsaw puzzles.
By a known reduction \cite{Jigsaw_GC}, these results imply NP-hardness
for polyomino packing (exact packing of a given set of
polyominoes into a given rectangle) when the polyominoes all have
area $\Theta(\log n)$; the previous NP-hardness proof \cite{Jigsaw_GC}
needed polyominoes of area $\Theta(\log^2 n)$.

We prove inapproximability for two different optimization versions of the
problems.  First, we consider placing the maximum number of tiles without
any violations of the matching constraints.
This objective has a simple $1 \over 2$-approximation for $1 \times n$
puzzles: alternate between placing a tile and leaving a blank.
(A $2 \over 3$-approximation is also possible; see Section~\ref{sec:open}.)
Second, we consider placing all of the tiles while maximizing the number of
compatible edges between adjacent tiles (as in \cite{Antoniadis-Lingas-2010}).%
\footnote{A dual objective would be to place all tiles while minimizing the
  number of mismatched edges, but this problem is already NP-hard to
  distinguish between an answer of zero and positive, so it cannot be
  approximated.}
This objective also has a simple $1 \over 2$-approximation for $1 \times n$
puzzles, via a maximum-cardinality matching on the tiles in a graph where
edges represent having any compatible edges: any solution with $k$
compatibilities induces a matching of size at least $k/2$
\cite{Antoniadis-Lingas-2010}.
Thus, up to constant factors, we resolve the approximability of these puzzles.

Our reduction is from Hamiltonian path on directed graphs whose
vertices each has maximum in-degree and out-degree 2, which was shown to be NP-complete by Plesn{\'i}k \cite{plesnik}.
To prove inapproximability, we reduce from a maximization version of this
Hamiltonicity problem, called \emph{maximum vertex-disjoint path cover}, where
the goal is to choose as many edges as possible to form vertex-disjoint paths.
(A Hamiltonian path would be an ideal path cover, forming a single
path of length $|V|-1$.)  This problem has a $12 \over 17$-approximation
\cite{Vishwanathan-1992},
and is NP-hard to approximate within some constant factor \cite{Enge03}
via a known connection to Asymmetric TSP with weights in $\{1,2\}$
\cite{Vishwanathan-1992}.

We prove that maximum vertex-disjoint path cover satisfies a stronger type
of hardness, called \emph{gap hardness}: it is NP-hard to distinguish between a
directed graph having a Hamiltonian path versus one where all vertex-disjoint
path covers having at most $ 0.999999284 \, |V|$
($> {1396639 \over 1396640} \, |V|$)
edges, given a promise that the graph falls into one of these two categories.
This gap hardness immediately implies inapproximability within a factor of
$0.999999284$ (though this constant is weaker than the known
inapproximability bound \cite{Enge03}).
More useful is that our reduction to $1 \times n$ jigsaw/edge-matching puzzles
is gap-preserving, implying gap hardness and inapproximability for the latter.
This approach lets us focus on ``perfect'' instances (where all tiles are
compatible) versus ``very bad'' instances (where many tiles are incompatible),
which seems far easier than standard L-reductions used in many
inapproximability results, where we must distinguish between an arbitrary
optimal and a factor below that arbitrary optimal.
We posit that gap hardness---where the high end of the gap is ``perfect'',
attaining the maximum possible bound, and thus matching the NP-hardness of the
original decision problem---is the better way to prove hardness of approximation
for many puzzles and games,
and hope that our results and approach find use in other research as well.



\section{Problem Definitions}

In this section, we formally define the relevant concepts and problems we
consider in this paper.

\subsection{Approximation Algorithms and Gap Problems}

One approach to dealing with NP-hard problems is to allow an algorithm to approximate the answer to the problem instead of solving it exactly. The quality of such an \emph{approximation algorithm} is given by its \emph{approximation ratio} as defined below. For the purposes of this paper, we define approximation ratio and subsequent notions only for maximization problems; these notions can be extended for minimization problems as well.

\begin{definition}
An approximation algorithm for a maximization problem has an \emph{approximation ratio} $\rho \leq 1$ if it outputs a solution that has value at least $\rho$ times the optimum.
We call such an algorithm a \emph{$\rho$-approximation algorithm}.
\end{definition}

Whereas many NP-hard problems admit polynomial-time $\rho$-approximation algorithms for some sensible $\rho$'s, some are proven to be NP-hard to approximate to within some approximation ratio. To prove such inapproximability results, one typically reduces from a known NP-hard problem (e.g., \textsc{3SAT}) to a maximization problem so that, if the original instance is a \textsc{yes} instance, then the optimum of the resulting instance is large and, if the original instance is a \text{no} instance, then the optimum is small. This polynomial-time reduction implies that the optimization is NP-hard to approximate to within the ratio between the optimums in the two cases. For convenience, we use the following notation of \emph{gap problems} to describe these inapproximability results.

\begin{definition}
For a maximization problem $P$ and $\beta > \gamma$, \Gap{$P$}[$\beta, \gamma$] is the problem of distinguishing whether an instance of $P$ has optimum at least $\beta$ or at most $\gamma$.
\end{definition}

Gap problems are widely used in hardness of approximation because NP-hardness of \Gap{$P$}[$\beta, \gamma$] (which we refer to informally as \emph{gap hardness}) implies NP-hardness of approximating $P$ to within a factor of $\gamma/\beta$.

\subsection{Edge-Matching and Jigsaw Puzzles}

Next, we give formal definitions of our main problems of interest, (unsigned) edge-matching and signed edge-matching puzzles (or jigsaw puzzles), as described in the introduction.
While our paper will focus on the $1 \times n$ rectangular grid case,
we define the more general problem of edge-matching puzzles on an $h \times w$
rectangular grid.
We begin with the unsigned decision problem:

\begin{definition}
	\textsc{$h \times w$ Edge-Matching Puzzle} is the following problem:

	\textsc{Input:} $n = h w$ unit-square tiles where each of the four sides has a color.

	\textsc{Output:} whether the tiles can be assembled into the $h \times w$ rectangular grid (board) so that any two adjacent pieces have the same color on the shared edge.
\end{definition}

\begin{figure}
	\centering
	\begin{subfigure}[c]{\textwidth}
		\centering
		\scalebox{0.8}{\input{2x3-feasible_paper.tikz}}
	\end{subfigure}
	\caption{An example of a \textsc{$2 \times 3$ Edge-Matching Puzzle} instance consisting of $6$ tiles (left), where all tiles can be assembled into a $2 \times 3$ rectangular grid, with matching colors on the edges of adjacent tiles (right). Colors are specified redundantly as numbers.}
	\label{fig:2x3-feasible}
\end{figure}
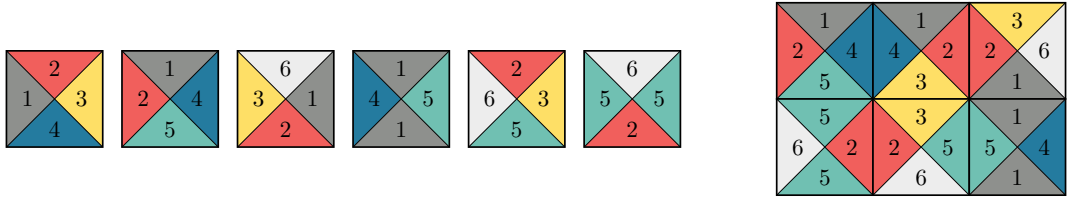

Figure~\ref{fig:2x3-feasible} shows an example of a \textsc{$2 \times 3$ Edge-Matching Puzzle}.
We write $\mathcal{T}(c_1, c_2, c_3, c_4)$ to denote a tile with colors $c_1, c_2, c_3, c_4$ on its left, upper, right, and lower edges respectively.
For example, the first tile in Figure~\ref{fig:2x3-feasible} is $\mathcal{T}(1,2,3,4)$.

In an assembly, we allow tiles to be rotated but not reflected.
It is easy to see from the proofs that all of our results hold
if we allow reflections in addition to rotations.
On the other hand, if rotations are prohibited, then \textsc{$1 \times n$ Edge-Matching Puzzle} is in P~\cite{EFW11}: with this restriction, the problem becomes equivalent to finding an Eulerian path (a path that visits every edge exactly once) in the directed graph whose vertices represent the colors and whose edges represent the tiles, which can be solved in linear time.

Next, we define two optimization versions of the edge-matching problem.
In the first version, the goal is to put as many tiles in the grid without any shared edge being labeled with different colors on two tiles.
In the second version (previously considered in \cite{Antoniadis-Lingas-2010}),
we must place all $n$ tiles into the rectangular grid, and the goal is to maximize the total number of matching edges between adjacent tiles.

\begin{definition}
	\textsc{$h \times w$ Max-Placement Edge-Matching Puzzle} is the following problem:

	\textsc{Input:} $n = h w$ unit-square tiles where each of the four sides has a color.

	\textsc{Output:} the maximum number of the tiles that can be assembled\footnote{To avoid confusion, we clarify that each tile must be placed parallel to the board's borders with integer-coordinated corners, i.e., the board has $h w$ unit-square slots and each tile is either put in a slot or left out of the board.} into the $h \times w$ rectangular grid such that any two adjacent pieces have the same color on the shared edge.
\end{definition}

\begin{definition} \label{def-mme-variation}
	\textsc{$h \times w$ Max-Matched Edge-Matching Puzzle} is the following problem:

	\textsc{Input:} $n = h w$ unit-square tiles where each of the four sides has a color and a sign on it.

	\textsc{Output:} the maximum number of same-color edges between adjacent tiles obtained by assembling all $n$ tiles into $h \times w$ rectangular grid.
\end{definition}

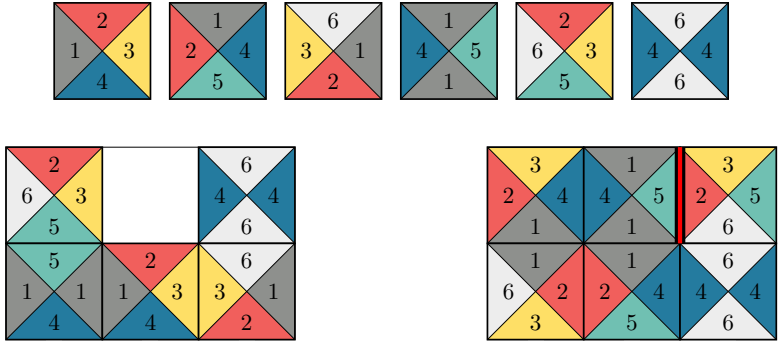
\begin{figure}
	\centering
	\begin{subfigure}[c]{\textwidth}
		\centering
		\scalebox{0.8}{\input{2x3-approx_paper.tikz}}
	\end{subfigure}
	\caption{An instance (top) corresponding to both a \textsc{$2 \times 3$ Max-Placement Edge-Matching Puzzle} whose optimal tiling contains $5$ tiles (bottom left), and a \textsc{$2 \times 3$ Max-Matched Edge-Matching Puzzle} whose optimal tiling yields $6$ matched edges (bottom right); the mismatched edge is indicated with a thick line.}
	\label{fig:2x3-approx}
\end{figure}

Figure~\ref{fig:2x3-approx} shows an example comparing both optimization problems.
In this paper, we focus on proving inapproximability of $1 \times n$ Max-Placement Edge-Matching Puzzle, then later show how to modify the proof to establish hardness under the Max-Matched variation.

Next we define the \emph{signed} variation of the puzzle, where each edge has a sign and the signs on the shared edge of two pieces must be opposite.

\begin{definition}
	\textsc{$h \times w$ Signed Edge-Matching Puzzle} is the following problem:

	\textsc{Input:} $n = h w$ unit-square tiles where each of the four sides has a color and a sign ($+$ or $-$).

	\textsc{Output:} whether the tiles can be assembled into the $h \times w$ rectangular grid so that any two adjacent pieces have the same color but opposite signs on the shared edge.
\end{definition}

Likewise, we use $\mathcal{T}(s_1c_1, s_2c_2, s_3c_3, s_4c_4)$ to denote a tile with colors $c_1, c_2, c_3, c_4$ and signs $s_1, s_2, s_3, s_4$ on its left, upper, right, and lower edges respectively.

Signed edge-matching puzzles can be viewed as jigsaw puzzles: $+$ and $-$ signs correspond to tabs and pockets respectively. It is worth noting, however, that the definition above is slightly different from typical jigsaw puzzles: signed edge-matching puzzles allow tabs and pocket on the borders, whereas jigsaw puzzles normally forbid this. If the borders of the puzzle must be flat (have no color or sign), then $1 \times n$ jigsaw puzzles become solvable in polynomial time in a similar fashion to edge-matching puzzles without rotations.

For signed edge-matching puzzles, we can define an optimization version similarly:
\begin{definition}
	\textsc{$h \times w$ Max-Placement Signed Edge-Matching Puzzle} is the following problem:

	\textsc{Input:} $n = h w$ unit-square tiles where each of the four sides has a color and a sign on it.

	\textsc{Output:} the maximum number of the tiles that can be assembled into the $h \times w$ rectangular grid such that any two adjacent pieces have the same color and different signs on the shared edge.
\end{definition}

It is possible to define the Max-Matched variation for signed
edge-matching puzzles, and our inapproximability result still holds.
Note, however, that this variation is unnatural when interpreted as a
jigsaw puzzle, because it would be physically difficult to place mismatched
jigsaw pieces adjacent to each other.


\subsection{Hamiltonian Path and Maximum Vertex-Disjoint Path Cover}

The Hamiltonian cycle problem is to decide whether a given (directed) graph has a \emph{Hamiltonian cycle}---a cycle that visits each vertex in the graph exactly once. This problem has long been known to be NP-complete; it was on Karp's initial list of 21 NP-complete problems~\cite{Karp72}. In this paper, we are more interested in its close relative, the Hamiltonian path problem:

\begin{definition}
  \textsc{Hamiltonian Path} is the following problem:

  \textsc{Input:} a directed graph $G = (V, E)$, a source $s$, and a sink $t$.\footnote{A \emph{source} of a directed graph is a vertex with in-degree zero; a \emph{sink} is a vertex with out-degree zero.}

  \textsc{Output:} whether there exists a Hamiltonian path (a path visiting each vertex exactly once) starting at $s$ and ending at~$t$.
\end{definition}

Like the Hamiltonian cycle problem, the Hamiltonian path problem has long been known to be NP-complete~\cite{GJ79}.%
\footnote{The conventional definition of \textsc{Hamiltonian Path} (e.g., in \cite{GJ79}) does not require $s$ and $t$ to be a source and a sink, respectively. We may easily reduce the conventional problem to our variant by adding a new starting vertex $s'$ with a single incident edge from $s'$ to $s$, and a new target vertex $t'$ with a single incident edge from $t$ to $t'$. Any Hamiltonian path in one instance can be converted into a corresponding Hamiltonian path in the other instance by adding or removing these new vertices and edges, establishing NP-hardness result for our variant.}
Garey, Johnson, and Tarjan \cite{GJT76} showed that the problem remains NP-hard even on undirected graphs where each vertex has degree at most three.
Plesn{\'i}k \cite{plesnik} extended the result to planar directed graphs with bounded in- and out-degrees.

While we do not use Plesn{\'i}k's result directly, we use his reduction from 3SAT to prove gap hardness of \textsc{$1 \times n$ Max-Placement Edge-Matching Puzzle}. To do so, we need an optimization version of the Hamiltonian path problem:

\begin{definition}
	\textsc{Max Vertex-Disjoint Path Cover($d$)} is the following problem:

	\textsc{Input:} a directed graph $G = (V, E)$ with a unique source $s$ and a unique sink $t$ such that the in-degree and out-degree of each vertex are at most $d$.

	\textsc{Output:} the maximum number of edges in any vertex-disjoint path cover---a collection of paths such that each vertex appears in exactly one path. Here we allow paths of length zero (i.e., a single vertex).
\end{definition}

\textsc{Max Vertex-Disjoint Path Cover($d$)} can be viewed as an optimization version of the Hamiltonian path problem because a graph has a vertex-disjoint path cover with $|V| - 1$ edges if and only if there exists a Hamiltonian path. The uniquenesses of source and sink and the degree requirement are imposed for technical reasons, which will be clear once our reduction is presented. 

Engebretsen \cite{Enge03} proved that Asymmetric TSP with weights in $\{1,2\}$
is NP-hard to approximate within a factor of $\frac{2805}{2804}-\varepsilon$,
which implies
\cite{Vishwanathan-1992} that \textsc{Max Vertex-Disjoint Path Cover($d$)}
is NP-hard to approximate within some constant factor (better than the
constant we will obtain).  Their proof implies a gap hardness result,
but not one where the high end of the gap has perfect solutions
(optimal value $|V|-1$, i.e., Hamiltonian), at least not without modification.
To fix this problem, and for self-containedness, we analyze our own
construction here.

\subsection{\textsc{Max-3SAT}}

To establish gap hardness for \textsc{Max Vertex-Disjoint Path Cover($d$)}, we resort to a classic problem: \textsc{Max-3SAT}.
Recall that a Boolean formula is in \emph{conjunctive normal form (CNF)} if it is a conjunction (\textsc{and}) of zero or more clauses, where each clause is a disjunction (\textsc{or}) of at most three literals, each of which is either a variable or its negation. For example, $(x_1 \vee x_2) \wedge (x_1 \vee \overline{x_3} \vee x_4) \wedge (\overline{x_2} \vee x_4)$ is in CNF with three clauses.

The celebrated Cook's Theorem (also known as Cook-Levin Theorem) states that it is NP-hard to decide whether a CNF formula is satisfiable~\cite{Cook71}. Shortly after, Karp \cite{Karp72} proved NP-completeness of the more restricted version \textsc{3SAT}, where each clause has at most three literals. Because we are aiming for an inapproximability result, we will use an optimization version of \textsc{3SAT}:

\begin{definition}
	\textsc{Max-3SAT($k$)} is the following problem:

	\textsc{Input:} a CNF formula such that each clause contains at most three literals and each variable appears in at most $k$ clauses.

	\textsc{Output:} the maximum number of clauses satisfied by any assignment.
\end{definition}

The PCP Theorem, considered a landmark in modern complexity theory, essentially states that \textsc{Max-3SAT} (without the bound $k$ on variable occurrences) is NP-hard to approximate to within some constant factor~\cite{ALMSS98,AS98}. \textsc{Max-3SAT($k$)} has also been researched intensively; even before the PCP Theorem was proven, Papadimitriou and Yannakakis \cite{PY88} showed that, for some constant $k$, \textsc{Max-3SAT($k$)} is a complete problem for a complexity class MaxSNP. Later, Feige \cite{Feige98} proved that \textsc{Max-3SAT(5)} is NP-hard to approximate up to some constant factor. For the purpose of this paper, we will use the following gap hardness of \textsc{Max-3SAT(29)} from \cite[pp.~314]{Vazirani01}, which has a more concrete ratio than that from~\cite{Feige98}.

\begin{lemma} \label{lem-3sat-approx}
For any constant $\alpha_{\SAT} < \frac{1}{344}$, \Gap{\textsc{Max-3SAT(29)}}$[m, (1-\alpha_{\SAT})m]$ is NP-hard.
\end{lemma}

\section{Our Contributions}

With formal problem statements in hand, we now formally state our results, starting with our main result for \textsc{$1 \times n$ Max-Placement Edge-Matching Puzzle}.

\begin{theorem} \label{thm:main}
For any nonnegative constant $\alpha_{\EMP} < \frac{1}{67038720}$, \Gap{\textsc{$1 \times n$ Max-Placement Edge-Matching Puzzle}}$[n, (1-\alpha_{\EMP})n]$ is NP-hard. In particular, it is NP-hard to approximate \textsc{$1 \times n$ Max-Placement Edge-Matching Puzzle} to within a factor of $\frac{67038719}{67038720}  + \delta$ of optimal for any sufficiently small constant $\delta > 0$.
\end{theorem}

Although NP-hardness of \textsc{$h \times w$ Edge-Matching Puzzle} is known for every $h \geq 1$~\cite{EFW11}, we are not aware of any existing inapproximability result for edge-matching puzzles. Our proof for the above theorem also yields an analogous result for \textsc{$1 \times n$ Max-Placement Signed Edge-Matching Puzzle} as stated below.

\begin{corollary} \label{cor:signed}
For any nonnegative constant $\alpha_{\EMP} < \frac{1}{67038720}$, \Gap{\textsc{$1 \times n$ Max-Placement Signed Edge-Matching Puzzle}}$[n, (1-\alpha_{\EMP})n]$ is NP-hard. In particular, it is NP-hard to approximate \textsc{$1 \times n$ Max-Placement Signed Edge-Matching Puzzle} to within a factor of $\frac{67038719}{67038720}  + \delta$ of optimal for any sufficiently small constant $\delta > 0$.
\end{corollary}

In addition, we prove a similar inapproximability result for the variation considered in \cite{Antoniadis-Lingas-2010} as stated below. Note that this result also holds for its corresponding signed variation by simply applying the proofs of both corollaries.

\begin{corollary} \label{cor:mme}
For any nonnegative constant $\alpha_{\EMP} < \frac{1}{67038720}$, \Gap{\textsc{$1 \times n$ Max-Matched Edge-Matching Puzzle}}$[n, (1-\alpha_{\EMP})n]$ is NP-hard. In particular, it is NP-hard to approximate \textsc{$1 \times n$ Max-Matched Edge-Matching Puzzle} to within a factor of $\frac{67038719}{67038720}  + \delta$ of optimal for any sufficiently small constant $\delta > 0$.
\end{corollary}

As an intermediate step to showing the results above, we prove gap hardness for \textsc{Max Vertex-Disjoint Path Cover(2)}:

\begin{theorem} \label{thm:gapmvdpc}
For any nonnegative constant $\alpha_{\MPC} < \frac{1}{1396640}$, \Gap{\textsc{Max Vertex-Disjoint Path Cover(2)}}$[|V|-1, (1-\alpha_{\MPC})(|V|-1)]$ is NP-hard. In particular, it is NP-hard to approximate \textsc{Max Vertex-Disjoint Path Cover(2)} to within a factor of $\frac{1396639}{1396640}  + \delta$ of optimal for any sufficiently small constant $\delta > 0$.
\end{theorem}


The outline for the rest of the paper is as follows. In Section~\ref{sec:exactemp}, we prove NP-hardness of \textsc{$1\times n$ Signed Edge-Matching Puzzle} by a reduction from \textsc{Hamiltonian Path}; this is not one of the main results described above, yet its simple proof provides good intuition for the subsequent inapproximability proofs. Next, in Section~\ref{sec:emp}, we prove the inapproximability results for $1 \times n$ edge-matching puzzles, Theorem~\ref{thm:main} and Corollaries~\ref{cor:signed}--\ref{cor:mme}, using a similar reduction from \textsc{Max Vertex-Disjoint Path Cover(2)}, an intermediate problem in achieving our results. We then prove the inapproximability result for \textsc{Max Vertex-Disjoint Path Cover(2)} in Section~\ref{sec:mvdpc}. Finally, we discuss further potential directions for research for the edge-matching problem in Section~\ref{sec:open}.

\section{NP-hardness of $1 \times n$ Signed Edge-Matching Puzzle} \label{sec:exactemp}

In this section we give a simple proof of NP-hardness of \textsc{$1 \times n$ Signed Edge-Matching Puzzle}. We provide a polynomial-time reduction from \textsc{Hamiltonian Path} to \textsc{$1 \times n$ Signed Edge-Matching Puzzle} in Section~\ref{sec:red-semp}. We then prove in Section~\ref{sec:red-semp-proof} that our reduction satisfies the following lemma:

\begin{lemma} \label{lem:exactsempred}
The polynomial-time reduction from \textsc{Hamiltonian Path} to \textsc{$1 \times n$ Signed Edge-Matching Puzzle} described in Section~\ref{sec:red-semp} is such that the \textsc{Hamiltonian Path} instance contains a Hamiltonian path if and only if the tiles of the constructed \textsc{$1 \times n$ Signed Edge-Matching Puzzle} instance can be assembled into an $1\times n$ grid without any mismatched edges.
\end{lemma}

This lemma and NP-hardness of \textsc{Hamiltonian Path} imply our desired hardness result:

\begin{corollary} \label{lem:exactsemp}
\textsc{$1 \times n$ Signed Edge-Matching Puzzle} is NP-hard.
\end{corollary}

\subsection{Construction of \textsc{$1 \times n$ Signed Edge-Matching Puzzle} Instance from \textsc{Hamiltonian Path} Instance} \label{sec:red-semp}

The overall strategy of the reduction from \textsc{Hamiltonian Path} to \textsc{$1 \times n$ Signed Edge-Matching Puzzle} is to represent each vertex and each edge by tiles, encoding the path in the ordering of the tiles. The leftover edge tiles will be packed at the end of the tiling. The path and the leftover edges in our tiling are then joined by a bridge tile.

\begin{figure}
	\centering
	\begin{subfigure}[c]{\textwidth}
		\centering
		\input{graph2.tikz}
		\caption{An example of a \textsc{Hamiltonian Path} instance, a directed graph $G$ with source $s=v_1$ and sink $t=v_5$. There exists a Hamiltonian path from $s$ to $t$, shown with thick edges.}
		\label{fig:graph2}
		\vspace{12pt}
	\end{subfigure}
	\begin{subfigure}[c]{\textwidth}
		\centering
		\scalebox{0.8}{\input{graph-overlaid2.tikz}}
		\caption{The constructed \textsc{$1 \times n$ Signed Edge-Matching Puzzle} instance (with $n=12$). On the graph structure are the vertex tiles (light gray) and the edge tiles (white). The only bridge tile (dark gray) is given on the right.}
		\label{fig:graph-overlaid2}
		\vspace{12pt}
	\end{subfigure}
	\begin{subfigure}[c]{\textwidth}
		\centering
		\scalebox{0.8}{\input{exampletiling2_paper.tikz}}
		\caption{The tiling constructed from the Hamiltonian path above. The path is represented in the first nine slots; observe that each slot contains a tile along the path in the above figure. The tenth slot contains the bridge tile, whereas the last two are filled with unused edge tiles sharing garbage-colored edges.}
		\label{fig:tiling2}
	\end{subfigure}

	\caption{An example of the reduction described in Section~\ref{sec:red-semp}, together with a Hamiltonian path of the \textsc{Hamiltonian Path} instance, and the corresponding tiling without any mismatched edges in the constructed \textsc{$1\times n$ Signed Edge-Matching Puzzle} instance.}
	\label{fig:reduction}
\end{figure}
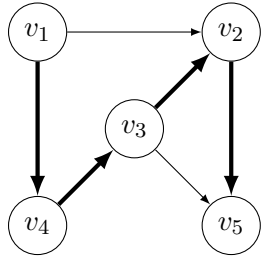
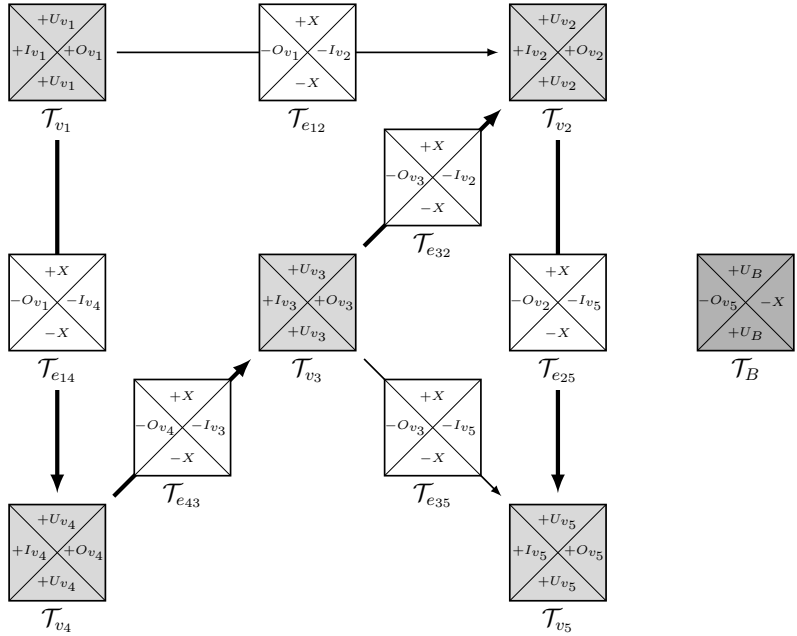
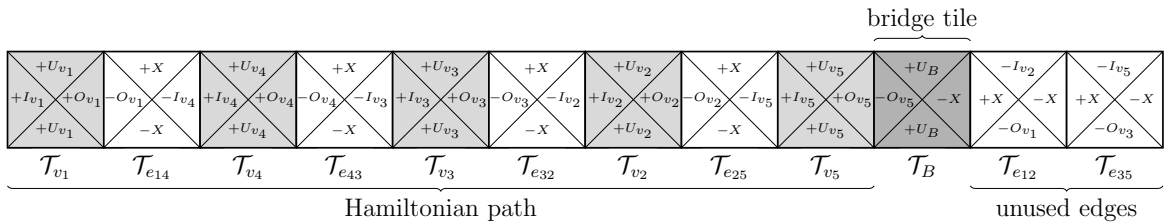

The reduction, illustrated in Figure~\ref{fig:graph2}--\ref{fig:graph-overlaid2}, takes an instance of \textsc{Hamiltonian Path}, which is a directed graph $G = (V, E)$ with a source $s$ and a sink $t$, then produces an instance of \textsc{$1 \times n$ Signed Edge-Matching Puzzle} as follows:
\begin{itemize}
	\item For each vertex $v_i \in V$, we add a \emph{vertex tile} $\mathcal{T}_{v_i}$ with positive \emph{vertex colors} $+I_{v_i}$ and $+O_{v_i}$ on opposite edges (each represents entering and leaving $v_i$, respectively), and the unique unmatchable color $+U_{v_i}$ on the remaining edges; i.e., $\mathcal{T}_{v_i} = \mathcal{T}(+I_{v_i}, +U_{v_i}, +O_{v_i}, +U_{v_i})$.
	\item For each edge $e_{ij}\in E$ from $v_i$ to $v_j$, we add an \emph{edge tile} $\mathcal{T}_{e_{ij}}$ with negative vertex colors $-O_{v_i}$ and $-I_{v_j}$ on opposite edges (each represents leaving $v_i$ and entering $v_j$, respectively), and the matching \emph{garbage colors} pair $+X$ and $-X$ on the two remaining edges. In other words, $\mathcal{T}_{e_{ij}} = \mathcal{T}(-O_{v_i}, +X, -I_{v_j}, -X)$.
	\item Finally, we have one \emph{bridge tile} $\mathcal{T}_B$ with the negative vertex color $-O_t$ on one edge, one garbage color $-X$ on the opposite edge, and a unique unmatchable color $+U_B$ on the remaining edges. That is, $\mathcal{T}_B=\mathcal{T}(-O_t,+U_B,-X,+U_B)$.
\end{itemize}

\subsection{Proof of Lemma~\ref{lem:exactsempred}} \label{sec:red-semp-proof}

Our reduction is clearly a polynomial-time reduction. We must now show that the existence of a Hamiltonian path in the \textsc{Hamiltonian Path} instance implies the existence of a tiling scheme without any mismatched edges in the constructed \textsc{$1 \times n$ Signed Edge-Matching Puzzle} instance, and vice versa.

\subsubsection*{\textsc{Hamiltonian Path} $\implies $\textsc{$1 \times n$ Signed Edge-Matching Puzzle}}

Turning a Hamiltonian path into a tiling scheme is easy; an example of this process can be found in Figure~\ref{fig:tiling2}. We encode our path into the leftmost $2|V| - 1$ slots in the grid by placing the tiles corresponding to the vertices and edges in the path alternately in the grid. More specifically, let the Hamiltonian path be $v_{\pi(1)}, \dots, v_{\pi(|V|)}$ where $v_{\pi(1)}=s$ and $v_{\pi(|V|)}=t$. For $j = 1, \dots, |V|$, we place the vertex tile $\mathcal{T}_{v_{j}}$ oriented as $\mathcal{T}(+I_{v_{\pi(j)}}, X, +O_{v_{\pi(j)}}, X)$ in the $(2j - 1)^\textrm{th}$ slot. For $j=1, \ldots, |V|-1$, we place the edge tile $\mathcal{T}_{e_{\pi(j)\pi(j+1)}}$ oriented as $\mathcal{T}(-O_{v_{\pi(j)}}, +X, -I_{v_{\pi(j+1)}}, -X)$ in the $(2j)^\textrm{th}$ slot.

We then place the bridge tile oriented as $\mathcal{T}_B=\mathcal{T}(-O_t,+U_B,-X,+U_B)$ in the $(2|V|)^\textrm{th}$ slot; since $t$ is the last vertex in the Hamiltonian path, this tile is compatible with the last tile in the $(2|V|-1)^\textrm{th}$ slot. For the remaining edge tiles, we rotate them as $\mathcal{T}(+X, -I_{v_j}, -X, -O_{v_i})$ so that the matching garbage colors are on the left and right, then put them into the rest of the slots. Clearly, we have placed every tile on the board without any conflict and, hence, we have completed the first half of the proof.

\subsubsection*{\textsc{$1 \times n$ Signed Edge-Matching Puzzle} $\implies$ \textsc{Hamiltonian Path}}

Assume that there exists a tiling scheme on a $1 \times n$ grid without any mismatched edges. Observe that the bridge tile must be oriented so that its unique color $+U_B$ is on its upper and lower edges because this color cannot be matched by any tile. Without loss of generality, suppose that the bridge tile is oriented as $\mathcal{T}_B=\mathcal{T}(-O_t,+U_B,-X,+U_B)$. (Otherwise, rotate the whole grid $180^{\circ}$.)

Consider first the tiles on the right side of the bridge tile. By our construction, we may only place edge tiles tile oriented as $\mathcal{T}(+X, -I_{v_j}, -X,-O_{v_i})$ in order to match their garbage-colored edges. Consequently, none of the vertex tiles may be placed; all of them must appear on the left side of the bridge tile.

Now consider the tiles on the left side of the bridge tile. Observe that, to be compatible with the color $-O_{v_j}$ to the right, we may only place the vertex tile $\mathcal{T}_{v_j}$ oriented as $\mathcal{T}(+I_{v_j}, +U_{v_j}, +O_{v_j}, +U_{v_j})$. Similarly, to be compatible with the color $+I_{v_j}$ to the right, we may only place some edge tile $\mathcal{T}_{e_{ij}}$ oriented as $\mathcal{T}(-O_{v_i}, +X, -I_{v_j}, -X)$. Therefore, these tiles may only be placed if they correspond to a path on $G$ ending at $t$. Since all vertex tiles appear to the left of the bridge tile, we have a path visiting all vertices ending at $t$. Moreover, as our starting vertex $s$ is a source, the color $+I_{s}$ cannot be matched by any tile. Thus $\mathcal{T}_s$ must be the left-most tile on the grid. That is, the tiles to the left of the bridge tile encode a Hamiltonian path of $G$, as desired.

\section{Inapproximability of $1 \times n$ Edge-Matching Puzzles} \label{sec:emp}

In this section, we generalize the approach from Section~\ref{sec:exactemp} in order to prove Theorem~\ref{thm:main}, NP-hardness of approximation of \textsc{$1 \times n$ Max-Placement Edge-Matching Puzzle}. To this end, we use the optimization variant of \textsc{Hamiltonian Path} called \textsc{Max Vertex-Disjoint Path Cover} defined earlier, and provide a gap-preserving reduction from \textsc{Max Vertex-Disjoint Path Cover} to \textsc{$1 \times n$ Max-Placement Edge-Matching Puzzle}. We largely focus on the unsigned variant because it is the more complicated case. Specifically, by removing the signs in the reduction of Section~\ref{sec:red-semp}, it becomes possible to place two edge tiles next to each other without using the garbage color since they share an additional color if they have the same starting vertex or ending vertex. We circumvent this problem by restricting to graphs of in-degrees and out-degrees at most two. This requries us to later show the hardness of approximation for \textsc{Max Vertex-Disjoint Path Cover(2)} in this more restricted family of graphs.

We describe our reduction from \textsc{Max Vertex-Disjoint Path Cover(2)} to \textsc{$1\times n$ Max-Placement Edge-Matching Puzzle} in Section~\ref{sec:red-memp}. We then prove in Section~\ref{sec:red-memp-proof} that the reduction satisfies the properties stated in the lemma below:

\begin{lemma} \label{lem-emp-approx}
For any nonnegative constants $\alpha_{\MPC}$ and $\alpha_{\EMP} < \alpha_{\MPC} / 48$, the following properties hold for the reduction described in Section~\ref{sec:red-memp} when $|V|$ is sufficiently large:
\begin{itemize}
	\item if the optimum of the \textsc{Max Vertex-Disjoint Path Cover(2)} instance is $|V| - 1$ (the graph contains a Hamiltonian path), then the optimum of the resulting \textsc{$1 \times n$ Max-Placement Edge-Matching Puzzle} instance is at least $n$, i.e., every tile can be placed compatibly on the board, and,
	\item if the optimum of the \textsc{Max Vertex-Disjoint Path Cover(2)} instance is at most $(1-\alpha_{\MPC})(|V| - 1)$, then the optimum of the resulting \textsc{$1 \times n$ Max-Placement Edge-Matching Puzzle} is at most $(1-\alpha_{\EMP})n$.
\end{itemize}
\end{lemma}

We will later show NP-hardness of approximation of \textsc{Max Vertex-Disjoint Path Cover(2)}, namely Theorem~\ref{thm:gapmvdpc}, in Section~\ref{sec:mvdpc}. The above lemma and Theorem~\ref{thm:gapmvdpc} immediately imply Theorem~\ref{thm:main}.

Then, we describe how to modify the reduction to arrive at a similar inapproximability result for \textsc{$1 \times n$ Max-Placement Signed Edge-Matching Puzzle} (Corollary~\ref{cor:signed}) in Section~\ref{sec:msemp}.
Lastly, we provide another reduction to obtain the inapproximability result for \textsc{$1 \times n$ Max-Matched Edge-Matching Puzzle} (Corollary~\ref{cor:mme}) in Section~\ref{sec:maxmatched}.

\subsection{Construction of \textsc{$1 \times n$ Max-Placement Edge-Matching Puzzle} Instance from \textsc{Max Vertex-Disjoint Path Cover(2)} Instance} \label{sec:red-memp}

The overall strategy of the reduction from \textsc{Max Vertex-Disjoint Path Cover(2)} to \textsc{$1 \times n$ Max-Placement Edge-Matching Puzzle} remains unchanged from Section~\ref{sec:red-semp}, except that the signs on the tiles are now removed, as shown in Figure~\ref{fig:tiles}. For completeness, we include a concise specification below:
\begin{itemize}
	\item For each vertex $v_i \in V$, add a vertex tile $\mathcal{T}_{v_i} = \mathcal{T}(I_{v_i}, U_{v_i}, O_{v_i}, U_{v_i})$.
	\item For each vertex $e_{ij} \in E$, add an edge tile $\mathcal{T}_{e_{ij}} = \mathcal{T}(O_{v_i}, X, I_{v_j}, X)$.
	\item Add a bridge tile $\mathcal{T}_B = \mathcal{T}(O_t, U_B, X, U_B)$.
\end{itemize}

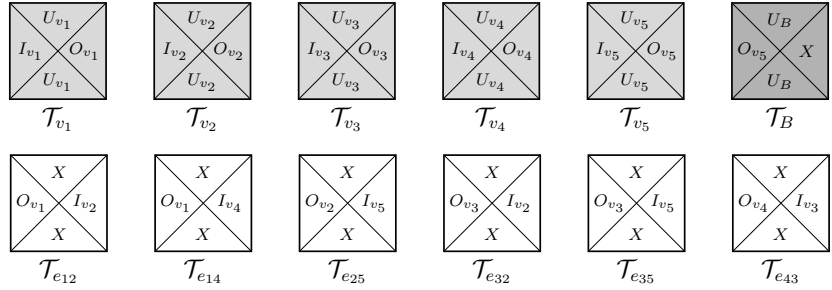
\begin{figure}[h!]
	\centering

	\begin{subfigure}[c]{\textwidth}
		\centering
		\scalebox{0.8}{\input{othervertex_paper.tikz}}
	\end{subfigure}
	\vspace{4pt}

	\begin{subfigure}[c]{\textwidth}
		\centering
		\scalebox{0.8}{\input{edges_paper.tikz}}
	\end{subfigure}

	\caption{The tiles of the \textsc{$1 \times n$ Max-Placement Edge-Matching Puzzle} instance created by the reduction in Section~\ref{sec:red-memp}, where the \textsc{Max Vertex-Disjoint Path Cover(2)} instance is the same graph from Figure~\ref{fig:graph2}.}
	\label{fig:tiles}
\end{figure}

Turning a vertex-disjoint path cover into a tiling scheme is still straightforward. We can simply encode each path by placing the tiles corresponding to the vertices and edges in the path alternately in the grid, leaving spaces between paths as necessary. Finally, we arrange the remaining edge tiles together by sharing garbage-colored edges. Intuitively, if there are only few paths in the path cover, then the number of blank spaces on the board is also small. In particular, if there is a Hamiltonian path, then all tiles can still be placed without any mismatched edges. 

On the other hand, converting a tiling configuration to a path cover is not as easy. Ideally, we want every edge tile $\mathcal{T}_{e_{ij}}$ that is packed in the grid to fall into one of the following categories.
\begin{itemize}
\item It is oriented so that the vertex-colored edges are on the left and the right ($\mathcal{T}(O_{v_i}, X, I_{v_j}, X)$ or $\mathcal{T}(I_{v_j}, X, O_{v_i}, X)$) and is placed between two vertex tiles. In this case, the corresponding edge should be included in the path cover.
\item It is oriented so that the garbage colors are on the left and right (i.e., $\mathcal{T}(X, I_{v_j}, X, O_{v_i})$ or $\mathcal{T}(X, O_{v_i}, X, I_{v_j})$). These edges are discarded from the path cover.
\end{itemize}

Unfortunately, it is possible for two edge tiles to be placed next to each other without the shared border being labeled with the garbage color; for instance, if there are edges $e_{12}$ and $e_{13}$, then we can rotate their corresponding tiles to be $\mathcal{T}(O_{v_2}, X, I_{v_1}, X)$ and $\mathcal{T}(I_{v_1}, X, O_{v_3}, X)$, and placed them next to each other. This is where the bounded degree constraint comes in. Suppose that each vertex in the graph has in-degree and out-degree at most two. If two edge tiles are placed consecutively with vertex color on the shared edge, there must be an empty slot, blank space or the board's border next to the vertex tile corresponding to that color (if the vertex tile is placed on the board). As a result, if there are not too many empty slots, we can reason that there must also be a small number of such pairs of edge tiles. We can then remove these problematic edge tiles and arrive at a path cover as desired. These arguments are formalized below.

\subsection{Proof of Lemma~\ref{lem-emp-approx}} \label{sec:red-memp-proof}

We will prove each half of Lemma~\ref{lem-emp-approx}.

\subsubsection*{\textsc{Max Vertex-Disjoint Path Cover(2)} $\implies $\textsc{$1 \times n$ Max-Placement Edge-Matching Puzzle}}



This proof remains unchanged from that case of \textsc{Hamiltonian Path} $\implies $\textsc{$1 \times n$ Signed Edge-Matching Puzzle} in Section~\ref{sec:red-semp-proof}, as removing the signs only makes the puzzle easier to solve.

%

\subsubsection*{\textsc{$1 \times n$ Max-Placement Edge-Matching Puzzle} $\implies$ \textsc{Max Vertex-Disjoint Path Cover(2)}}

We will prove this by contrapositive. Suppose that the resulting \textsc{$1 \times n$ Max-Placement Edge-Matching Puzzle} instance has optimum more than $(1 - \alpha_{\EMP})n$, i.e., there is a tiling with less than $\alpha_{\EMP} n$ empty slots. We translate the tiling to a path cover as follows:
\begin{enumerate}
  \item Remove the bridge tile from the board. \label{step:bri}
	\item Consider all edge tiles that are rotated in such a way that left and right edges are vertex-colored. Among these tiles, remove the ones that are adjacent to a blank space, a (left or right) border of the board or another edge tile. \label{step:eli}
	\item After the previous step, two consecutive edge tiles can only share a garbage-colored edge. Moreover, from the colors we choose for vertex tiles, two vertex tiles cannot be placed next to each other. As a result, each connected component (a maximal set of contiguous nonempty slots) either (1) alternates between vertex and edge tiles, or (2) consists solely of edge tiles that share only garbage-colored edges.

	Because we removed all edge tiles with left and right vertex-colored edges that are next to a blank space or a board's border, each sequence of the first type must start and end with vertex tiles, which represents a path. We create a path cover consisting of all paths corresponding to such sequences, and length-zero paths, one for each vertex tile not in the board.
\end{enumerate}

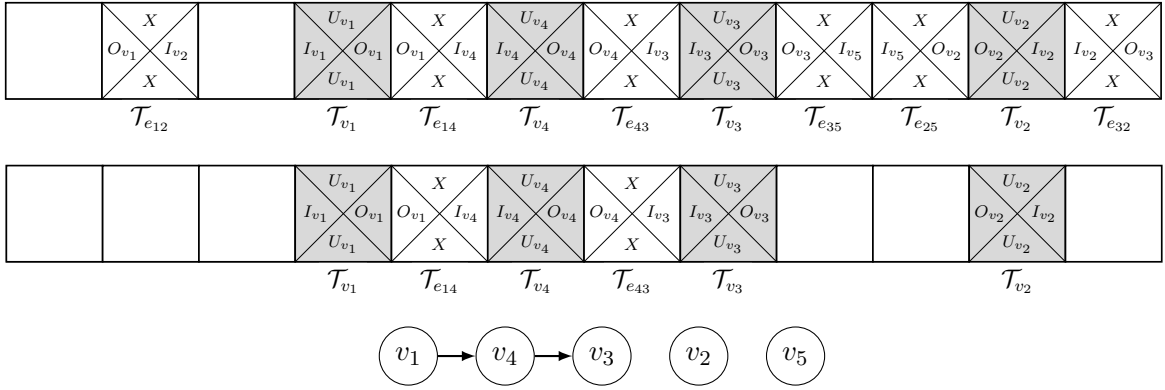
\begin{figure}
	\centering
	\begin{subfigure}[c]{\textwidth}
		\centering
		\scalebox{0.8}{\input{exampletiling-algo_paper.tikz}}
	\end{subfigure}
	\vspace{8pt}

	\centering
	\begin{subfigure}[c]{\textwidth}
		\centering
		\scalebox{0.8}{\input{exampletiling-algo-eliminated_paper.tikz}}
	\end{subfigure}
	\vspace{4pt}

	\centering
	\begin{subfigure}[c]{\textwidth}
		\centering
		\input{path-cover.tikz}
	\end{subfigure}

	\caption{An example of how the algorithm works for a tiling scheme for the instance from Figure~\ref{fig:tiles}. The upper tiling is the input. Since the bridge tile is not put on the board, the tiling is not modified in Step~\ref{step:bri}. In Step~\ref{step:eli}, $\mathcal{T}_{e_{12}}$ and $\mathcal{T}_{e_{32}}$ are removed because they are next to a blank space and a board border respectively; $\mathcal{T}_{e_{35}}$ and $\mathcal{T}_{e_{25}}$ are removed because they share a vertex-colored edge. This results in the lower tiling. The third step of the algorithm then picks two paths ($v_1 \rightarrow v_4 \rightarrow v_3$ and $v_2$) from the board and one ($v_5$) from an unused vertex tile.}
	\label{fig:tiling-algo}
\end{figure}

Figure~\ref{fig:tiling-algo} shows the algorithm in action for an example tiling scheme. Clearly, the resulting collection of paths is a vertex-disjoint path cover. We will now argue that this path cover indeed contains more than $(|V| - 1)(1 - \alpha_{\MPC})$ edges. We start by proving the following useful lemma.

\begin{lemma} \label{lem:cleared-edge-tile}
Less than $14\alpha_{\EMP}n + 12$ edge tiles are removed in Step~\ref{step:eli}.
\end{lemma}

\begin{proof}
After the bridge tile is removed, there are less than $\alpha_{\EMP}n + 1$ empty slots. This means that, in Step~\ref{step:eli}, less than $2(\alpha_{\EMP}n + 1) + 2 = 2\alpha_{\EMP}n + 4$ edge tiles are removed because they are adjacent to a blank space or a border of the board. The rest must be removed because they share a vertex color-labeled edge with other edge tiles.

Suppose that the shared edge of two consecutive edge tiles is labeled by $I_{v_i}$ or $O_{v_i}$ for some vertex ${v_i}$. Because $v_i$ has in-degree and out-degree at most two, the tile $\mathcal{T}_{v_i}$ must be either omitted from the board or is next to a blank space or a border of the board. There can be at most $\alpha_{\EMP}n + 2\alpha_{\EMP}n + 2 = 3\alpha_{\EMP}n + 2$ such $\mathcal{T}_{v_i}$'s. For each such $\mathcal{T}_{v_i}$'s, there are only four edge tiles that have $I_{v_i}$ or $O_{v_i}$. Hence, the number of edge tiles removed by such cause is at most $4(3\alpha_{\EMP}n + 2) = 12\alpha_{\EMP}n + 8$.

Thus, the total number of edge tiles removed is less than $14\alpha_{\EMP}n + 12$.
\end{proof}

Next, recall that the number of edges in the path cover is $|V|$ minus the number of paths. There are two type of paths we create: (1) paths from tiling in the board and (2) paths from unused vertex tiles. The number of the latter is bounded by the number of initially empty slots, which is less than $\alpha_{\EMP}n$. The number of the former is at most the number of connected components in the board after Step~\ref{step:eli}, which is bounded by one plus the number of empty slots.

The empty slots on the board are either initially empty or are made empty in the first two steps. Recall that the number of initially empty slots is less than $\alpha_{\EMP}n$. One additional empty slot occurs in Step~\ref{step:bri}. Finally, from Lemma~\ref{lem:cleared-edge-tile}, the number of empty slots created in Step~\ref{step:eli} is less than $14\alpha_{\EMP}n + 12$. Hence, the total number of empty slots is less than $15\alpha_{\EMP}n + 13$. Thus, the number of edges in the path cover is more than
\begin{align*}
|V| - (1 + 15\alpha_{\EMP}n + 13) - \alpha_{\EMP}n
&= |V| - 14 - 16\alpha_{\EMP}(|V| + |E| + 1) \\
&\geq |V| - 14 - 16\alpha_{\EMP}(|V| + 2|V| + 1) \\
&\geq (1 - 48\alpha_{\EMP})|V| - 30.
\end{align*}

The first inequality follows from the fact that each vertex has in-degree at most two.

Finally, since $\alpha_{\EMP} < \alpha_{\MPC}/48$, when $|V|$ is sufficiently large, the last quantity is more than $(1 - \alpha_{\MPC})(|V| - 1)$. We have thus completed the proof of Lemma~\ref{lem-emp-approx}.

\subsection{Proof of Corollary~\ref{cor:signed} (Signed Variation)} \label{sec:msemp}

Because the proof of Corollary~\ref{cor:signed} largely resembles that of Theorem~\ref{thm:main}, we will not repeat the whole proof here. The main difference is in the reduction: we instead need to follow the reduction from Section~\ref{sec:red-semp} by restoring the omitted signs. Apart from this, the rest of the proof remains unchanged.


\subsection{Proof of Corollary~\ref{cor:mme} (Maximum-Matched Variation)} \label{sec:maxmatched}
In this section, we provide a gap reduction to the maximum-matched variation, although this time our instance does not need to be modified. More specifically, we prove the following lemma.

\begin{lemma} \label{lem-maxmatched-approx}
For any constant $\alpha < 1$, the following properties hold:
\begin{itemize}
	\item if the optimum of \textsc{$1 \times n$ Max-Placement Edge-Matching Puzzle} is $n$, then the optimum of  \textsc{$1 \times n$ Max-Matched Edge-Matching Puzzle} on the same instance is $n-1$, and,
	\item if the optimum of \textsc{$1 \times n$ Max-Placement Edge-Matching Puzzle} is at most $(1-\alpha)n$, then the optimum of  \textsc{$1 \times n$ Max-Matched Edge-Matching Puzzle} on the same instance is at most $(1-\alpha)(n-1)$.
\end{itemize}
\end{lemma}

By picking $\alpha = \alpha_{\EMP}$ from Theorem~\ref{thm:main}, Lemma~\ref{lem-maxmatched-approx} together with Theorem~\ref{thm:main} immediately implies NP-hardness of \Gap{\textsc{$1 \times n$ Max-Matched Edge-Matching Puzzle}}$[n-1, (1-\alpha_{\EMP})(n-1)]$, establishing Corollary~\ref{cor:mme}.

\begin{proof}[Proof of Lemma~\ref{lem-maxmatched-approx}]
For clarity, we refer to the \textsc{$1 \times n$ Max-Placement Edge-Matching Puzzle} problem, where the goal is to maximize the number of matched edges, as the original problem.

\subsubsection*{Original problem $\implies$ Maximum-Matched Variation}
This is straightforward since a perfect tiling contains $n - 1$ matched edges.

\subsubsection*{Maximum-Matched Variation $\implies$ Original problem}
We again prove by contrapositive; suppose that the maximum matched edges variation has a solution with more than $(1-\alpha)(n-1)$ matched edges. That is, there are less than $\alpha(n-1) < \alpha n$ mismatched edges. We create a solution to the original problem by removing all the tiles whose right edges are mismatched. Because each mismatched edge contributes to at most one empty slot, there are less than $\alpha n$ blank spaces, i.e., the number of tiles is more than $(1 - \alpha) n$.
\end{proof}

We note that our proof holds for the signed variation, and so does our inapproximability result.

\section{Inapproximability of \textsc{Max Vertex-Disjoint Path Cover(2)}} \label{sec:mvdpc}

We now prove the inapproximability of \textsc{Max Vertex-Disjoint Path Cover(2)} via a reduction from \textsc{Max-3SAT(29)}. Throughout section, let $n$ and $m$ denote the number of variables and clauses in the CNF instance of \textsc{Max-3SAT(29)}, respectively.

Our proof of the inapproximability of \textsc{Max Vertex-Disjoint Path Cover(2)} primarily relies on the reduction used to prove NP-hardness of the Hamiltonian cycle problem on graphs with degree bound two in \cite{plesnik}. We modify Plesn{\'i}k's construction to create a gap-preserving reduction, then apply our reduction on \textsc{Max-3SAT(29)}. Section~\ref{plesnik-construct} addresses the construction in detail. This construction turns a CNF instance of \textsc{Max-3SAT(29)} into a graph instance $G=(V, E)$ of \textsc{Max Vertex-Disjoint Path Cover(2)} satisfying the following lemma.

\begin{lemma} \label{lem-mvdpc-approx}
For any nonnegative constants $\alpha_{\SAT}$ and $\alpha_{MPC} = \alpha_{\SAT}/4060$, the following properties hold for the reduction described in Section~\ref{plesnik-construct}:
\begin{itemize}
	\item if the optimum of the \textsc{Max-3SAT(29)} instance is $m$ (the corresponding CNF formula is satisfiable), then the optimum of the resulting \textsc{Max Vertex-Disjoint Path Cover(2)} instance is $|V|-1$ (the corresponding graph contains a Hamiltonian path), and,
	\item if the optimum of the \textsc{Max-3SAT(29)} instance is at most $(1-\alpha_{\SAT})m$, then the optimum of the resulting \textsc{Max Vertex-Disjoint Path Cover(2)} is at most $(1-\alpha_{\MPC})(|V|-1)$.
\end{itemize}
\end{lemma}

Our proof of Lemma~\ref{lem-mvdpc-approx} is provided in Section~\ref{proof-mvdpc-approx}.
Lemma~\ref{lem-mvdpc-approx} and Lemma~\ref{lem-3sat-approx} immediately imply the NP-hardness for the gap problem of \textsc{Max Vertex-Disjoint Path Cover(2)} described in Theorem~\ref{thm:gapmvdpc}.

\subsection{Construction of \textsc{Max Vertex-Disjoint Path Cover(2)} Instance from \textsc{Max-3SAT(29)} Instance} \label{plesnik-construct}

Assume that we are given a \textsc{Max-3SAT(29)} instance $\phi$; namely, $\phi$ consists of $n$ variables $x_1, \ldots, x_n$, and $m$ clauses $c_1, \ldots, c_m$ where each $c_j$ is a disjunction of (at most) $3$ literals $y_{j1}, y_{j2}, y_{j3} \in \{x_1, \overline{x_1},\ldots, x_n, \overline{x_n}\}$, such that no variable appears as literals more than $29$ times. Using Plesn{\'i}k's construction, we similarly define the clause gadget and the variable gadget as shown in Figure~\ref{fig:gadgets}; each gadget is enclosed in a rounded dashed rectangle.\footnote{We contracted some edges in Plesn{\'i}k's gadgets; this allows us to achieve a slightly better value of $\alpha_{\MPC}$.} Each variable gadget has two \emph{literal edges} corresponding to $x_i$ and $\overline{x_i}$, respectively. Figure \ref{fig:connect} shows how we connect our gadgets together. Unlike Plesn{\'i}k's construction, we create a source node $s$ and a target node $t$ rather than creating an edge between the top vertices of the first clause and variable gadgets.

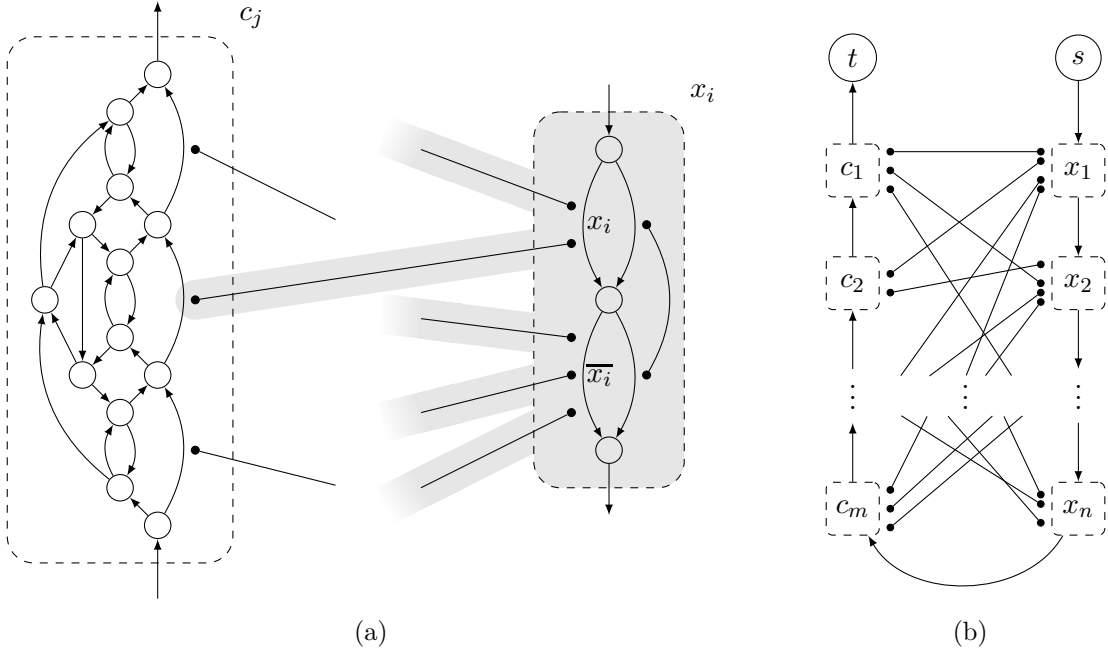
\begin{figure}
    \centering
    \begin{subfigure}[b]{0.65\textwidth}
    	\centering
        \input{gadgets.tikz}
        \caption{}
        \label{fig:gadgets}
    \end{subfigure}
    \begin{subfigure}[b]{0.3\textwidth}
    	\centering
        \input{connect.tikz}
        \caption{}
        \label{fig:connect}
    \end{subfigure}
    \caption{(a) A clause gadget for $c_j$ (left) and a variable gadget for $x_i$ (right), where $c_j$ contains $x_i$ as its second literal; the shaded area indicates $x_i$'s variable territory. (b) Full construction of our \textsc{Max Vertex-Disjoint Path Cover(2)} instance.}
\end{figure}

Each arc with black circular endpoints is a shorthand notation for an \emph{exclusive-or line (XOR line)}. The XOR gadget, which is the realization of an XOR line, is given in Figures \ref{fig:xor-short}-\ref{fig:xor-real}. There is one XOR line within a variable gadget on the right, and an XOR line in the middle connecting a literal edge to each occurrence of that literal in each clause. In Figure~\ref{fig:gadgets}, for example, clause $c_j$ contains a three literals, the middle of which is $y_{j2} = x_i$.

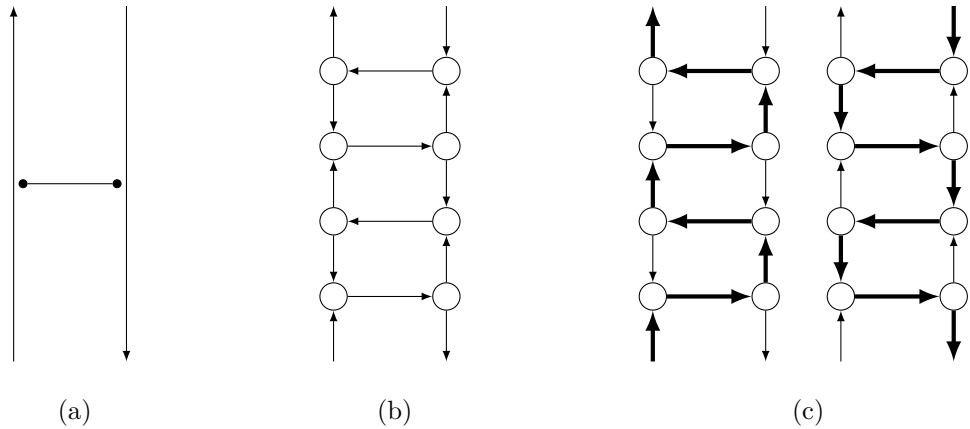
\begin{figure}
    \centering
    \begin{subfigure}[b]{0.25\textwidth}
    	\centering
        \input{xor-short.tikz}
        \caption{}
        \label{fig:xor-short}
    \end{subfigure}
    \begin{subfigure}[b]{0.25\textwidth}
    	\centering
        \input{xor-real.tikz}
        \caption{}
        \label{fig:xor-real}
    \end{subfigure}
    \begin{subfigure}[b]{0.4\textwidth}
    	\centering
        \input{xor-noendpts.tikz}
        \caption{}
        \label{fig:xor-noendpts}
    \end{subfigure}
    \caption{(a)~The shorthand notation for an XOR line. (b)~The XOR gadget realizing the XOR line of~(a). (c)~The only two possible configurations of path cover (thick edges) without inducing endpoints within an XOR line.}
\end{figure}

For each variable $x_i$, we define its \emph{variable territory} as the set of vertices in its variable gadget, along with the eight vertices introduced by each of its XOR lines, including those on the clause gadgets where $x_i$ or $\overline{x_i}$ appear. This is illustrated as the shaded region on Figure~\ref{fig:gadgets}. We also define a \emph{clause territory} of each clause $c_j$ as the union between the vertices of $c_j$'s variable gadget, and the variable territories of variables whose literals appear in $c_j$.

\subsection{Proof of Lemma~\ref{lem-mvdpc-approx}} \label{proof-mvdpc-approx}

We will next prove each half of Lemma~\ref{lem-mvdpc-approx}.

\subsubsection*{\textsc{Max-3SAT(29)}$\implies$\textsc{Max Vertex-Disjoint Path Cover(2)}}

Suppose that the input \textsc{Max-3SAT(29)} instance is satisfiable. We need to show that the constructed \textsc{Max Vertex-Disjoint Path Cover(2)} instance contains a path cover of size $|V|-1$. Observe that there is a path cover of size $|V|-1$ if and only if the path cover consists entirely of a single (Hamiltonian) path from $s$ to $t$.

Let us define \emph{endpoints} of a path cover as a set of the sources and the destinations of paths in the cover. Then, it is sufficient to give a path cover that has no other endpoints besides $s$ and $t$. First, we establish the following observation which can be easily checked via a simple case analysis.

\begin{observation} \label{obs-xor}
Consider Figure~\ref{fig:xor-noendpts}. If a path cover has no endpoints on any of the eight vertices introduced by an XOR line, then exactly one edge of the XOR line (in the shorthand notation) is in the path cover. More precisely, in the actual realization, either the pair of edges entering and leaving the XOR gadget on the left, or the pair on the right, is in the path cover.
\end{observation}

Now we are ready to explicitly provide the path cover; this is largely the same as that from in \cite{plesnik}, but included here for completeness. We begin by adding the literal edges corresponding to the optimal (satisfying) assignment of the \textsc{Max-3SAT(29)}; that is, if $x_i = \textsc{True}$ we add the literal edge $x_i$ to the cover; otherwise we add the literal edge $\overline{x_i}$. Due to Observation~\ref{obs-xor} above, the edges in the variable gadgets in our cover are fully determined; an example in case $x_i$ is set to \textsc{True} is given in Figure~\ref{fig:assignment}. Observe that the path starting from $s$ visits all vertices of all variable gadgets, and, at the same time, restricts that the edge on the opposite end of the XOR line is traversed if and only if the corresponding literal is \textsc{False}.

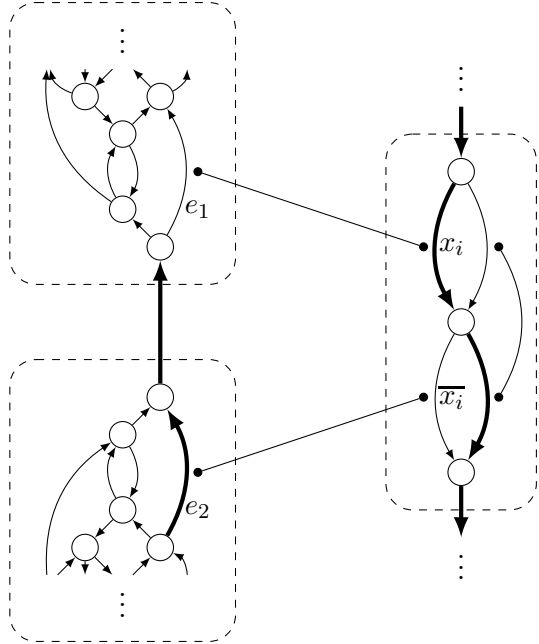
\begin{figure}
    \centering
    \input{assignment.tikz}
    \caption{Thick edges indicate edges that must be in the path cover when $x_i$ is set to \textsc{True}. Edge $e_1$ corresponds to a \textsc{True} literal ($x_i$), so it must not be traversed when we visit the clause side of the graph; edge $e_2$ corresponds to a \textsc{False} literal ($\overline{x_i}$) and therefore must be traversed.}
    \label{fig:assignment}
\end{figure}

Now consider our path when it reaches the clause side of the graph. As shown in Figure~\ref{fig:satisfy}, if the clause is satisfied, some edges on the right of the clause gadget does not need to be traversed. By a simple case analysis, we can cover all unvisited vertices on the clause gadget without breaking up our path. Thus, we have constructed a path cover of size $|V|-1$, as needed.

\begin{figure}
    \centering
    \begin{subfigure}[b]{0.6\textwidth}
    	\centering
        \input{satisfy.tikz}
        \caption{}
        \label{fig:satisfy}
    \end{subfigure}
    \begin{subfigure}[b]{0.3\textwidth}
    	\centering
        \input{notsatisfy.tikz}
        \caption{}
        \label{fig:notsatisfy}
    \end{subfigure}
    \caption{(a)~Possible constructions of path covers when the clause is satisfied (symmetric cases are omitted). (b)~The shaded region must contain an endpoint if the path cover includes all thick edges.}
\end{figure}
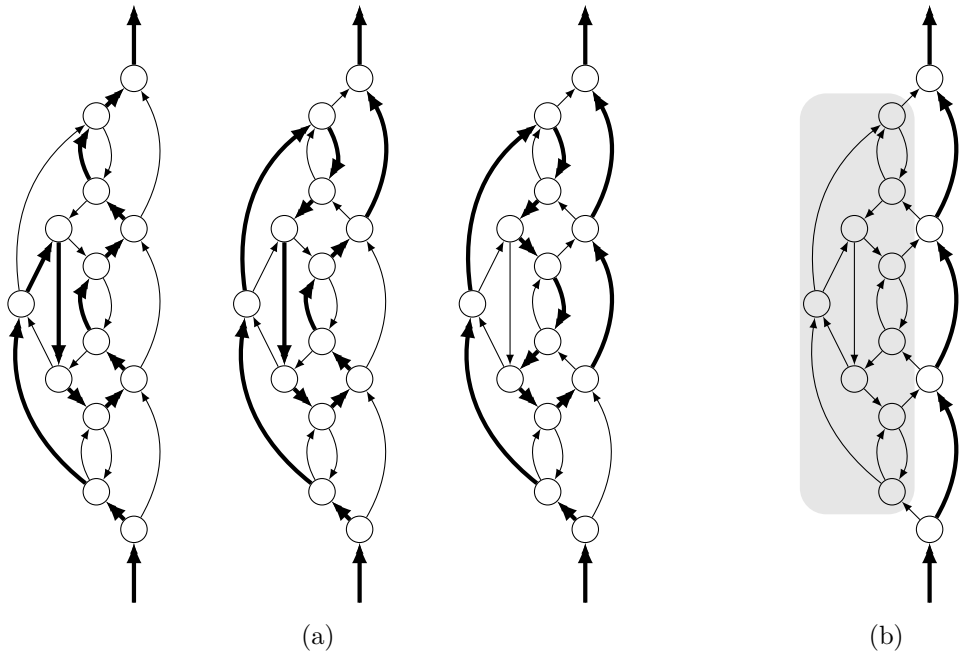

\subsubsection*{\textsc{Max Vertex-Disjoint Path Cover(2)}$\implies$\textsc{Max-3SAT(29)}}

We prove that second part of Lemma~\ref{lem-mvdpc-approx} via contrapositive. That is, suppose that the optimum of the constructed \textsc{Max Vertex-Disjoint Path Cover(2)} instance is greater than $(1-\alpha_{\MPC})(|V|-1)$, we must show that the optimum of the original \textsc{Max-3SAT(29)} instance is greater than $(1-\alpha_{\SAT})m$.

Consider an optimal path cover of our constructed \textsc{Max Vertex-Disjoint Path Cover} instance. Our path cover is of size greater than $(1-\alpha_{\MPC})(|V|-1)$, so it must contain less than $\alpha_{\MPC}|V|+\alpha_{\MPC} \leq \alpha_{\MPC}|V| + 1$ paths.  That is, besides $s$ and $t$, our path cover contains strictly less than $2\alpha_{\MPC}|V|$ other endpoints. To complete our proof, we will explicitly construct an assignment satisfying more than $(1-\alpha_{\SAT})m$ clauses.

Now, observe that if a variable territory contains no endpoints, then exactly one of the two literal edges are in the path cover due to the XOR lines as argued in Observation~\ref{obs-xor}. We construct our assignment as follows. If there are no endpoints in $x_i$'s variable territory, we assign the truth value according to the literal in the path cover. Otherwise, assign an arbitrary one. We now show that our assignment fails to satisfy strictly less than $\alpha_{\SAT} m$ clauses.

First, we prove the following lemma.

\begin{lemma} \label{lemma-clause-endpoint}
If a clause $c_j$ is not satisfied by our assignment, then $c_j$'s clause territory contains an endpoint.
\end{lemma}
\begin{proof}
Suppose that there are no endpoints in any of the variable territories of any variable occurring in $c_j$. Then for $c_j$ to be unsatisfied, all literal edges in $c_j$ must be \textsc{False} and not included in the path cover. Observation~\ref{obs-xor} implies that we must traverse all three edges on the right of $c_j$'s clause gadget; these are shown in Figure~\ref{fig:notsatisfy} as the curved thick edges. Notice that our path leaves the shaded vertices disconnected from the remaining part of the graph, which requires at least one additional path to cover. That is, there must still be an endpoint in the clause territory of $c_j$.
\end{proof}

Lemma~\ref{lemma-clause-endpoint} implies that an endpoint must be present for a clause to be unsatisfied. Next, we show that each endpoint cannot be contained in very many clause territories.

\begin{lemma}
Each endpoint is in at most $29$ clause territories.
\end{lemma}
\begin{proof}
If the endpoint is not in any variable territory, then it is in at most one clause territory because clause gadgets do not overlap. Otherwise, notice that the endpoint can be only in one variable territory as variable territories also do not overlap. One variable may occur at most $29$ times, so a variable territory can be in at most $29$ clause territories. Therefore, any endpoint can be in at most $29$ clause territories.
\end{proof}

Each endpoint may cause up to $29$ clauses to become unsatisfied. Recall that there are less than $2\alpha_{\MPC}|V|$ endpoints that belong to some clause territories. Thus our assignment leaves strictly less than $29\cdot2\alpha_{\MPC}|V|=58\alpha_{\MPC}|V|$ clauses unsatisfied.

Finally, we relate $|V|$ and $m$ in order to establish our claim. Notice that we add vertices $s$ and $t$ to our graph for convenience when defining endpoints, but they can be safely removed as their only neighbors do not have other incoming or outgoing edges, respectively. Then, the number of vertices becomes $|V| \leq 11n + 37m \leq 11(3m) + 37m = 70m$. As the number of unsatisfied clauses is less than $58\alpha_{\MPC}(70m)=4060m$, our claim holds because $\alpha_{\MPC} = \alpha_{\SAT}/4060$, as desired.

\section{Conclusion and Open Problems} \label{sec:open}

Our results establish an upper bound on the optimal approximation factor at ${67038719 \over 67038720}$,
and the simple approximation algorithms mentioned
in the Introduction give a lower bound of $\frac{1}{2}$.  The optimal
approximation factor remains an open question, though the right answer is
probably closer to $\frac{1}{2}$.
We believe our inapproximability factor can be tightened somewhat by
further modifying Plesn\'ik's construction, or by reducing from a different problem; e.g.,
find a variation of \textsc{SAT} that yields a better trade-off in the reduction,
or make use of the tighter bounds on \textsc{$\{1,2\}$-TSP} from~\cite{Enge03}.


For the objective of placing the maximum tiles without violations,
we can improve the approximation factor for $1 \times n$ puzzles
to ${2 \over 3}$ by computing a maximum-cardinality matching on the tiles,
and interspersing matched groups (of size 2 or 1) with blanks.
Suppose that an optimal solution contains $OPT = \left\lceil {n \over 2} \right\rceil + k$ tiles.
Such a solution contains $n - (\left\lceil {n \over 2} \right\rceil + k) = \left\lfloor {n \over 2} \right\rfloor - k$ blank slots, leaving at least $(n-1) - 2(\left\lfloor {n \over 2} \right\rfloor - k) \geq 2k - 1$ matched edges, inducing a matching of cardinality at least $\left\lceil {{2k-1} \over 2} \right\rceil = k$ for these matched edges which form paths.
If $k \geq {n \over 3}$, this leads to an approximate solution consisting of pairs of matched tiles, except possibly for the last tile, packing at least ${2 \over 3}n \geq {2 \over 3}OPT$ tiles (because $OPT \leq n$).
Otherwise, the solution contains $2k$ tiles in the first $3k$ slots, ending with a blank slot, followed by alternation in the remaining $n - 3k$ spaces, for a total of $2k + \left\lceil{1 \over 2} (n - 3k)\right\rceil \geq \left\lceil {{n+k} \over 2} \right\rceil = \left\lceil {OPT \over 2} + {{\left\lfloor n/2 \right\rfloor} \over 2}\right\rceil \geq {2 \over 3}OPT$ tiles placed.





Another natural question is whether our inapproximability results also apply to
the original NP-hardness scenario of square target shapes \cite{Jigsaw_GC}.
We expect that our proofs can be adapted to this case, but it requires care to
avoid introducing additional error in the approximation factor.

Edge-matching puzzles are also common with tiles in the shape of regular
hexagons or equilateral triangles.  Are these problems are NP-hard and
inapproximable for ``$1 \times n$'' (but bumpy) target shapes?  Our results
carry over directly to hexagons, just by leaving two opposite edges of each
hexagon unlabeled, and using the other two pairs of opposite edges to simulate
a square tile.  Two triangles connected by a unique glue can simulate a square,
which suffices to prove NP-hardness, but more care would be required for
inapproximability to consider the case that the triangle pair gets separated.

\bibliographystyle{alpha}
\bibliography{citations}

\end{document}

%% file: 2x3-feasible_paper.tikz
\begin{tikzpicture}
\definecolor{color1}{rgb}{0.558594,0.566406,0.554688}
\definecolor{color2}{rgb}{0.945312,0.371094,0.359375}
\definecolor{color3}{rgb}{0.996094,0.875000,0.398438}
\definecolor{color4}{rgb}{0.140625,0.480469,0.625000}
\definecolor{color5}{rgb}{0.437500,0.753906,0.699219}
\definecolor{color6}{rgb}{0.929688,0.929688,0.929688}
\fill[fill=color1] (0.0, 0.0) -- (0.8, 0.8) -- (0.0, 1.6) -- cycle;
\fill[fill=color3] (1.6, 0.0) -- (0.8, 0.8) -- (1.6, 1.6) -- cycle;
\fill[fill=color2] (0.0, 1.6) -- (0.8, 0.8) -- (1.6, 1.6) -- cycle;
\fill[fill=color4] (0.0, 0.0) -- (0.8, 0.8) -- (1.6, 0.0) -- cycle;
\node [anchor=north,fill=white] at (0.8, 0.0) {\large $$};
\node () at (0.8, 0.8) [rectangle,draw,fill=none,thick,minimum width=1.6cm, minimum height=1.6cm] {};
\draw [color=black] (0.0, 0.0) -- (1.6, 1.6);
\draw [color=black] (0.0, 1.6) -- (1.6, 0.0);
\node [color=black] at (0.352, 0.8) {$1$};
\node [color=black] at (1.248, 0.8) {$3$};
\node [color=black] at (0.8, 0.288) {$4$};
\node [color=black] at (0.8, 1.312) {$2$};
\fill[fill=color2] (1.92, 0.0) -- (2.72, 0.8) -- (1.92, 1.6) -- cycle;
\fill[fill=color4] (3.52, 0.0) -- (2.72, 0.8) -- (3.52, 1.6) -- cycle;
\fill[fill=color1] (1.92, 1.6) -- (2.72, 0.8) -- (3.52, 1.6) -- cycle;
\fill[fill=color5] (1.92, 0.0) -- (2.72, 0.8) -- (3.52, 0.0) -- cycle;
\node [anchor=north,fill=white] at (2.72, 0.0) {\large $$};
\node () at (2.72, 0.8) [rectangle,draw,fill=none,thick,minimum width=1.6cm, minimum height=1.6cm] {};
\draw [color=black] (1.92, 0.0) -- (3.52, 1.6);
\draw [color=black] (1.92, 1.6) -- (3.52, 0.0);
\node [color=black] at (2.272, 0.8) {$2$};
\node [color=black] at (3.168, 0.8) {$4$};
\node [color=black] at (2.72, 0.288) {$5$};
\node [color=black] at (2.72, 1.312) {$1$};
\fill[fill=color3] (3.84, 0.0) -- (4.64, 0.8) -- (3.84, 1.6) -- cycle;
\fill[fill=color1] (5.44, 0.0) -- (4.64, 0.8) -- (5.44, 1.6) -- cycle;
\fill[fill=color6] (3.84, 1.6) -- (4.64, 0.8) -- (5.44, 1.6) -- cycle;
\fill[fill=color2] (3.84, 0.0) -- (4.64, 0.8) -- (5.44, 0.0) -- cycle;
\node [anchor=north,fill=white] at (4.64, 0.0) {\large $$};
\node () at (4.64, 0.8) [rectangle,draw,fill=none,thick,minimum width=1.6cm, minimum height=1.6cm] {};
\draw [color=black] (3.84, 0.0) -- (5.44, 1.6);
\draw [color=black] (3.84, 1.6) -- (5.44, 0.0);
\node [color=black] at (4.192, 0.8) {$3$};
\node [color=black] at (5.088, 0.8) {$1$};
\node [color=black] at (4.64, 0.288) {$2$};
\node [color=black] at (4.64, 1.312) {$6$};
\fill[fill=color4] (5.76, 0.0) -- (6.56, 0.8) -- (5.76, 1.6) -- cycle;
\fill[fill=color5] (7.36, 0.0) -- (6.56, 0.8) -- (7.36, 1.6) -- cycle;
\fill[fill=color1] (5.76, 1.6) -- (6.56, 0.8) -- (7.36, 1.6) -- cycle;
\fill[fill=color1] (5.76, 0.0) -- (6.56, 0.8) -- (7.36, 0.0) -- cycle;
\node [anchor=north,fill=white] at (6.56, 0.0) {\large $$};
\node () at (6.56, 0.8) [rectangle,draw,fill=none,thick,minimum width=1.6cm, minimum height=1.6cm] {};
\draw [color=black] (5.76, 0.0) -- (7.36, 1.6);
\draw [color=black] (5.76, 1.6) -- (7.36, 0.0);
\node [color=black] at (6.112, 0.8) {$4$};
\node [color=black] at (7.008, 0.8) {$5$};
\node [color=black] at (6.56, 0.288) {$1$};
\node [color=black] at (6.56, 1.312) {$1$};
\fill[fill=color6] (7.68, 0.0) -- (8.48, 0.8) -- (7.68, 1.6) -- cycle;
\fill[fill=color3] (9.28, 0.0) -- (8.48, 0.8) -- (9.28, 1.6) -- cycle;
\fill[fill=color2] (7.68, 1.6) -- (8.48, 0.8) -- (9.28, 1.6) -- cycle;
\fill[fill=color5] (7.68, 0.0) -- (8.48, 0.8) -- (9.28, 0.0) -- cycle;
\node [anchor=north,fill=white] at (8.48, 0.0) {\large $$};
\node () at (8.48, 0.8) [rectangle,draw,fill=none,thick,minimum width=1.6cm, minimum height=1.6cm] {};
\draw [color=black] (7.68, 0.0) -- (9.28, 1.6);
\draw [color=black] (7.68, 1.6) -- (9.28, 0.0);
\node [color=black] at (8.032, 0.8) {$6$};
\node [color=black] at (8.928, 0.8) {$3$};
\node [color=black] at (8.48, 0.288) {$5$};
\node [color=black] at (8.48, 1.312) {$2$};
\fill[fill=color5] (9.6, 0.0) -- (10.4, 0.8) -- (9.6, 1.6) -- cycle;
\fill[fill=color5] (11.2, 0.0) -- (10.4, 0.8) -- (11.2, 1.6) -- cycle;
\fill[fill=color6] (9.6, 1.6) -- (10.4, 0.8) -- (11.2, 1.6) -- cycle;
\fill[fill=color2] (9.6, 0.0) -- (10.4, 0.8) -- (11.2, 0.0) -- cycle;
\node [anchor=north,fill=white] at (10.4, 0.0) {\large $$};
\node () at (10.4, 0.8) [rectangle,draw,fill=none,thick,minimum width=1.6cm, minimum height=1.6cm] {};
\draw [color=black] (9.6, 0.0) -- (11.2, 1.6);
\draw [color=black] (9.6, 1.6) -- (11.2, 0.0);
\node [color=black] at (9.952, 0.8) {$5$};
\node [color=black] at (10.848, 0.8) {$5$};
\node [color=black] at (10.4, 0.288) {$2$};
\node [color=black] at (10.4, 1.312) {$6$};
\fill[fill=color2] (12.8, 0.8) -- (13.6, 1.6) -- (12.8, 2.4) -- cycle;
\fill[fill=color4] (14.4, 0.8) -- (13.6, 1.6) -- (14.4, 2.4) -- cycle;
\fill[fill=color1] (12.8, 2.4) -- (13.6, 1.6) -- (14.4, 2.4) -- cycle;
\fill[fill=color5] (12.8, 0.8) -- (13.6, 1.6) -- (14.4, 0.8) -- cycle;
\node [anchor=north,fill=white] at (13.6, 0.8) {\large $$};
\node () at (13.6, 1.6) [rectangle,draw,fill=none,thick,minimum width=1.6cm, minimum height=1.6cm] {};
\draw [color=black] (12.8, 0.8) -- (14.4, 2.4);
\draw [color=black] (12.8, 2.4) -- (14.4, 0.8);
\node [color=black] at (13.152, 1.6) {$2$};
\node [color=black] at (14.048, 1.6) {$4$};
\node [color=black] at (13.6, 1.088) {$5$};
\node [color=black] at (13.6, 2.112) {$1$};
\fill[fill=color4] (14.4, 0.8) -- (15.2, 1.6) -- (14.4, 2.4) -- cycle;
\fill[fill=color2] (16.0, 0.8) -- (15.2, 1.6) -- (16.0, 2.4) -- cycle;
\fill[fill=color1] (14.4, 2.4) -- (15.2, 1.6) -- (16.0, 2.4) -- cycle;
\fill[fill=color3] (14.4, 0.8) -- (15.2, 1.6) -- (16.0, 0.8) -- cycle;
\node [anchor=north,fill=white] at (15.2, 0.8) {\large $$};
\node () at (15.2, 1.6) [rectangle,draw,fill=none,thick,minimum width=1.6cm, minimum height=1.6cm] {};
\draw [color=black] (14.4, 0.8) -- (16.0, 2.4);
\draw [color=black] (14.4, 2.4) -- (16.0, 0.8);
\node [color=black] at (14.752, 1.6) {$4$};
\node [color=black] at (15.648, 1.6) {$2$};
\node [color=black] at (15.2, 1.088) {$3$};
\node [color=black] at (15.2, 2.112) {$1$};
\fill[fill=color2] (16.0, 0.8) -- (16.8, 1.6) -- (16.0, 2.4) -- cycle;
\fill[fill=color6] (17.6, 0.8) -- (16.8, 1.6) -- (17.6, 2.4) -- cycle;
\fill[fill=color3] (16.0, 2.4) -- (16.8, 1.6) -- (17.6, 2.4) -- cycle;
\fill[fill=color1] (16.0, 0.8) -- (16.8, 1.6) -- (17.6, 0.8) -- cycle;
\node [anchor=north,fill=white] at (16.8, 0.8) {\large $$};
\node () at (16.8, 1.6) [rectangle,draw,fill=none,thick,minimum width=1.6cm, minimum height=1.6cm] {};
\draw [color=black] (16.0, 0.8) -- (17.6, 2.4);
\draw [color=black] (16.0, 2.4) -- (17.6, 0.8);
\node [color=black] at (16.352, 1.6) {$2$};
\node [color=black] at (17.248, 1.6) {$6$};
\node [color=black] at (16.8, 1.088) {$1$};
\node [color=black] at (16.8, 2.112) {$3$};
\fill[fill=color6] (12.8, -0.8) -- (13.6, 0.0) -- (12.8, 0.8) -- cycle;
\fill[fill=color2] (14.4, -0.8) -- (13.6, 0.0) -- (14.4, 0.8) -- cycle;
\fill[fill=color5] (12.8, 0.8) -- (13.6, 0.0) -- (14.4, 0.8) -- cycle;
\fill[fill=color5] (12.8, -0.8) -- (13.6, 0.0) -- (14.4, -0.8) -- cycle;
\node [anchor=north,fill=white] at (13.6, -0.8) {\large $$};
\node () at (13.6, 0.0) [rectangle,draw,fill=none,thick,minimum width=1.6cm, minimum height=1.6cm] {};
\draw [color=black] (12.8, -0.8) -- (14.4, 0.8);
\draw [color=black] (12.8, 0.8) -- (14.4, -0.8);
\node [color=black] at (13.152, 0.0) {$6$};
\node [color=black] at (14.048, 0.0) {$2$};
\node [color=black] at (13.6, -0.512) {$5$};
\node [color=black] at (13.6, 0.512) {$5$};
\fill[fill=color2] (14.4, -0.8) -- (15.2, 0.0) -- (14.4, 0.8) -- cycle;
\fill[fill=color5] (16.0, -0.8) -- (15.2, 0.0) -- (16.0, 0.8) -- cycle;
\fill[fill=color3] (14.4, 0.8) -- (15.2, 0.0) -- (16.0, 0.8) -- cycle;
\fill[fill=color6] (14.4, -0.8) -- (15.2, 0.0) -- (16.0, -0.8) -- cycle;
\node [anchor=north,fill=white] at (15.2, -0.8) {\large $$};
\node () at (15.2, 0.0) [rectangle,draw,fill=none,thick,minimum width=1.6cm, minimum height=1.6cm] {};
\draw [color=black] (14.4, -0.8) -- (16.0, 0.8);
\draw [color=black] (14.4, 0.8) -- (16.0, -0.8);
\node [color=black] at (14.752, 0.0) {$2$};
\node [color=black] at (15.648, 0.0) {$5$};
\node [color=black] at (15.2, -0.512) {$6$};
\node [color=black] at (15.2, 0.512) {$3$};
\fill[fill=color5] (16.0, -0.8) -- (16.8, 0.0) -- (16.0, 0.8) -- cycle;
\fill[fill=color4] (17.6, -0.8) -- (16.8, 0.0) -- (17.6, 0.8) -- cycle;
\fill[fill=color1] (16.0, 0.8) -- (16.8, 0.0) -- (17.6, 0.8) -- cycle;
\fill[fill=color1] (16.0, -0.8) -- (16.8, 0.0) -- (17.6, -0.8) -- cycle;
\node [anchor=north,fill=white] at (16.8, -0.8) {\large $$};
\node () at (16.8, 0.0) [rectangle,draw,fill=none,thick,minimum width=1.6cm, minimum height=1.6cm] {};
\draw [color=black] (16.0, -0.8) -- (17.6, 0.8);
\draw [color=black] (16.0, 0.8) -- (17.6, -0.8);
\node [color=black] at (16.352, 0.0) {$5$};
\node [color=black] at (17.248, 0.0) {$4$};
\node [color=black] at (16.8, -0.512) {$1$};
\node [color=black] at (16.8, 0.512) {$1$};
\end{tikzpicture}

%% file: 2x3-approx_paper.tikz
\begin{tikzpicture}
\definecolor{color1}{rgb}{0.558594,0.566406,0.554688}
\definecolor{color2}{rgb}{0.945312,0.371094,0.359375}
\definecolor{color3}{rgb}{0.996094,0.875000,0.398438}
\definecolor{color4}{rgb}{0.140625,0.480469,0.625000}
\definecolor{color5}{rgb}{0.437500,0.753906,0.699219}
\definecolor{color6}{rgb}{0.929688,0.929688,0.929688}
\fill[fill=color1] (-5.6, 0.8) -- (-4.8, 1.6) -- (-5.6, 2.4) -- cycle;
\fill[fill=color3] (-4.0, 0.8) -- (-4.8, 1.6) -- (-4.0, 2.4) -- cycle;
\fill[fill=color2] (-5.6, 2.4) -- (-4.8, 1.6) -- (-4.0, 2.4) -- cycle;
\fill[fill=color4] (-5.6, 0.8) -- (-4.8, 1.6) -- (-4.0, 0.8) -- cycle;
\node [anchor=north,fill=white] at (-4.8, 0.8) {\large $$};
\node () at (-4.8, 1.6) [rectangle,draw,fill=none,thick,minimum width=1.6cm, minimum height=1.6cm] {};
\draw [color=black] (-5.6, 0.8) -- (-4.0, 2.4);
\draw [color=black] (-5.6, 2.4) -- (-4.0, 0.8);
\node [color=black] at (-5.248, 1.6) {$1$};
\node [color=black] at (-4.352, 1.6) {$3$};
\node [color=black] at (-4.8, 1.088) {$4$};
\node [color=black] at (-4.8, 2.112) {$2$};
\fill[fill=color2] (-3.68, 0.8) -- (-2.88, 1.6) -- (-3.68, 2.4) -- cycle;
\fill[fill=color4] (-2.08, 0.8) -- (-2.88, 1.6) -- (-2.08, 2.4) -- cycle;
\fill[fill=color1] (-3.68, 2.4) -- (-2.88, 1.6) -- (-2.08, 2.4) -- cycle;
\fill[fill=color5] (-3.68, 0.8) -- (-2.88, 1.6) -- (-2.08, 0.8) -- cycle;
\node [anchor=north,fill=white] at (-2.88, 0.8) {\large $$};
\node () at (-2.88, 1.6) [rectangle,draw,fill=none,thick,minimum width=1.6cm, minimum height=1.6cm] {};
\draw [color=black] (-3.68, 0.8) -- (-2.08, 2.4);
\draw [color=black] (-3.68, 2.4) -- (-2.08, 0.8);
\node [color=black] at (-3.328, 1.6) {$2$};
\node [color=black] at (-2.432, 1.6) {$4$};
\node [color=black] at (-2.88, 1.088) {$5$};
\node [color=black] at (-2.88, 2.112) {$1$};
\fill[fill=color3] (-1.76, 0.8) -- (-0.96, 1.6) -- (-1.76, 2.4) -- cycle;
\fill[fill=color1] (-0.16, 0.8) -- (-0.96, 1.6) -- (-0.16, 2.4) -- cycle;
\fill[fill=color6] (-1.76, 2.4) -- (-0.96, 1.6) -- (-0.16, 2.4) -- cycle;
\fill[fill=color2] (-1.76, 0.8) -- (-0.96, 1.6) -- (-0.16, 0.8) -- cycle;
\node [anchor=north,fill=white] at (-0.96, 0.8) {\large $$};
\node () at (-0.96, 1.6) [rectangle,draw,fill=none,thick,minimum width=1.6cm, minimum height=1.6cm] {};
\draw [color=black] (-1.76, 0.8) -- (-0.16, 2.4);
\draw [color=black] (-1.76, 2.4) -- (-0.16, 0.8);
\node [color=black] at (-1.408, 1.6) {$3$};
\node [color=black] at (-0.512, 1.6) {$1$};
\node [color=black] at (-0.96, 1.088) {$2$};
\node [color=black] at (-0.96, 2.112) {$6$};
\fill[fill=color4] (0.16, 0.8) -- (0.96, 1.6) -- (0.16, 2.4) -- cycle;
\fill[fill=color5] (1.76, 0.8) -- (0.96, 1.6) -- (1.76, 2.4) -- cycle;
\fill[fill=color1] (0.16, 2.4) -- (0.96, 1.6) -- (1.76, 2.4) -- cycle;
\fill[fill=color1] (0.16, 0.8) -- (0.96, 1.6) -- (1.76, 0.8) -- cycle;
\node [anchor=north,fill=white] at (0.96, 0.8) {\large $$};
\node () at (0.96, 1.6) [rectangle,draw,fill=none,thick,minimum width=1.6cm, minimum height=1.6cm] {};
\draw [color=black] (0.16, 0.8) -- (1.76, 2.4);
\draw [color=black] (0.16, 2.4) -- (1.76, 0.8);
\node [color=black] at (0.512, 1.6) {$4$};
\node [color=black] at (1.408, 1.6) {$5$};
\node [color=black] at (0.96, 1.088) {$1$};
\node [color=black] at (0.96, 2.112) {$1$};
\fill[fill=color6] (2.08, 0.8) -- (2.88, 1.6) -- (2.08, 2.4) -- cycle;
\fill[fill=color3] (3.68, 0.8) -- (2.88, 1.6) -- (3.68, 2.4) -- cycle;
\fill[fill=color2] (2.08, 2.4) -- (2.88, 1.6) -- (3.68, 2.4) -- cycle;
\fill[fill=color5] (2.08, 0.8) -- (2.88, 1.6) -- (3.68, 0.8) -- cycle;
\node [anchor=north,fill=white] at (2.88, 0.8) {\large $$};
\node () at (2.88, 1.6) [rectangle,draw,fill=none,thick,minimum width=1.6cm, minimum height=1.6cm] {};
\draw [color=black] (2.08, 0.8) -- (3.68, 2.4);
\draw [color=black] (2.08, 2.4) -- (3.68, 0.8);
\node [color=black] at (2.432, 1.6) {$6$};
\node [color=black] at (3.328, 1.6) {$3$};
\node [color=black] at (2.88, 1.088) {$5$};
\node [color=black] at (2.88, 2.112) {$2$};
\fill[fill=color4] (4.0, 0.8) -- (4.8, 1.6) -- (4.0, 2.4) -- cycle;
\fill[fill=color4] (5.6, 0.8) -- (4.8, 1.6) -- (5.6, 2.4) -- cycle;
\fill[fill=color6] (4.0, 2.4) -- (4.8, 1.6) -- (5.6, 2.4) -- cycle;
\fill[fill=color6] (4.0, 0.8) -- (4.8, 1.6) -- (5.6, 0.8) -- cycle;
\node [anchor=north,fill=white] at (4.8, 0.8) {\large $$};
\node () at (4.8, 1.6) [rectangle,draw,fill=none,thick,minimum width=1.6cm, minimum height=1.6cm] {};
\draw [color=black] (4.0, 0.8) -- (5.6, 2.4);
\draw [color=black] (4.0, 2.4) -- (5.6, 0.8);
\node [color=black] at (4.352, 1.6) {$4$};
\node [color=black] at (5.248, 1.6) {$4$};
\node [color=black] at (4.8, 1.088) {$6$};
\node [color=black] at (4.8, 2.112) {$6$};
\draw (-6.4,0.0) rectangle (-1.6,-3.2);
\fill[fill=color6] (-6.4, -1.6) -- (-5.6, -0.8) -- (-6.4, 0.0) -- cycle;
\fill[fill=color3] (-4.8, -1.6) -- (-5.6, -0.8) -- (-4.8, 0.0) -- cycle;
\fill[fill=color2] (-6.4, 0.0) -- (-5.6, -0.8) -- (-4.8, 0.0) -- cycle;
\fill[fill=color5] (-6.4, -1.6) -- (-5.6, -0.8) -- (-4.8, -1.6) -- cycle;
\node [anchor=north,fill=white] at (-5.6, -1.6) {\large $$};
\node () at (-5.6, -0.8) [rectangle,draw,fill=none,thick,minimum width=1.6cm, minimum height=1.6cm] {};
\draw [color=black] (-6.4, -1.6) -- (-4.8, 0.0);
\draw [color=black] (-6.4, 0.0) -- (-4.8, -1.6);
\node [color=black] at (-6.048, -0.8) {$6$};
\node [color=black] at (-5.152, -0.8) {$3$};
\node [color=black] at (-5.6, -1.312) {$5$};
\node [color=black] at (-5.6, -0.288) {$2$};
\fill[fill=color4] (-3.2, -1.6) -- (-2.4, -0.8) -- (-3.2, 0.0) -- cycle;
\fill[fill=color4] (-1.6, -1.6) -- (-2.4, -0.8) -- (-1.6, 0.0) -- cycle;
\fill[fill=color6] (-3.2, 0.0) -- (-2.4, -0.8) -- (-1.6, 0.0) -- cycle;
\fill[fill=color6] (-3.2, -1.6) -- (-2.4, -0.8) -- (-1.6, -1.6) -- cycle;
\node [anchor=north,fill=white] at (-2.4, -1.6) {\large $$};
\node () at (-2.4, -0.8) [rectangle,draw,fill=none,thick,minimum width=1.6cm, minimum height=1.6cm] {};
\draw [color=black] (-3.2, -1.6) -- (-1.6, 0.0);
\draw [color=black] (-3.2, 0.0) -- (-1.6, -1.6);
\node [color=black] at (-2.848, -0.8) {$4$};
\node [color=black] at (-1.952, -0.8) {$4$};
\node [color=black] at (-2.4, -1.312) {$6$};
\node [color=black] at (-2.4, -0.288) {$6$};
\fill[fill=color1] (-6.4, -3.2) -- (-5.6, -2.4) -- (-6.4, -1.6) -- cycle;
\fill[fill=color1] (-4.8, -3.2) -- (-5.6, -2.4) -- (-4.8, -1.6) -- cycle;
\fill[fill=color5] (-6.4, -1.6) -- (-5.6, -2.4) -- (-4.8, -1.6) -- cycle;
\fill[fill=color4] (-6.4, -3.2) -- (-5.6, -2.4) -- (-4.8, -3.2) -- cycle;
\node [anchor=north,fill=white] at (-5.6, -3.2) {\large $$};
\node () at (-5.6, -2.4) [rectangle,draw,fill=none,thick,minimum width=1.6cm, minimum height=1.6cm] {};
\draw [color=black] (-6.4, -3.2) -- (-4.8, -1.6);
\draw [color=black] (-6.4, -1.6) -- (-4.8, -3.2);
\node [color=black] at (-6.048, -2.4) {$1$};
\node [color=black] at (-5.152, -2.4) {$1$};
\node [color=black] at (-5.6, -2.912) {$4$};
\node [color=black] at (-5.6, -1.888) {$5$};
\fill[fill=color1] (-4.8, -3.2) -- (-4.0, -2.4) -- (-4.8, -1.6) -- cycle;
\fill[fill=color3] (-3.2, -3.2) -- (-4.0, -2.4) -- (-3.2, -1.6) -- cycle;
\fill[fill=color2] (-4.8, -1.6) -- (-4.0, -2.4) -- (-3.2, -1.6) -- cycle;
\fill[fill=color4] (-4.8, -3.2) -- (-4.0, -2.4) -- (-3.2, -3.2) -- cycle;
\node [anchor=north,fill=white] at (-4.0, -3.2) {\large $$};
\node () at (-4.0, -2.4) [rectangle,draw,fill=none,thick,minimum width=1.6cm, minimum height=1.6cm] {};
\draw [color=black] (-4.8, -3.2) -- (-3.2, -1.6);
\draw [color=black] (-4.8, -1.6) -- (-3.2, -3.2);
\node [color=black] at (-4.448, -2.4) {$1$};
\node [color=black] at (-3.552, -2.4) {$3$};
\node [color=black] at (-4.0, -2.912) {$4$};
\node [color=black] at (-4.0, -1.888) {$2$};
\fill[fill=color3] (-3.2, -3.2) -- (-2.4, -2.4) -- (-3.2, -1.6) -- cycle;
\fill[fill=color1] (-1.6, -3.2) -- (-2.4, -2.4) -- (-1.6, -1.6) -- cycle;
\fill[fill=color6] (-3.2, -1.6) -- (-2.4, -2.4) -- (-1.6, -1.6) -- cycle;
\fill[fill=color2] (-3.2, -3.2) -- (-2.4, -2.4) -- (-1.6, -3.2) -- cycle;
\node [anchor=north,fill=white] at (-2.4, -3.2) {\large $$};
\node () at (-2.4, -2.4) [rectangle,draw,fill=none,thick,minimum width=1.6cm, minimum height=1.6cm] {};
\draw [color=black] (-3.2, -3.2) -- (-1.6, -1.6);
\draw [color=black] (-3.2, -1.6) -- (-1.6, -3.2);
\node [color=black] at (-2.848, -2.4) {$3$};
\node [color=black] at (-1.952, -2.4) {$1$};
\node [color=black] at (-2.4, -2.912) {$2$};
\node [color=black] at (-2.4, -1.888) {$6$};
\draw (1.6,0.0) rectangle (6.4,-3.2);
\fill[fill=color2] (1.6, -1.6) -- (2.4, -0.8) -- (1.6, 0.0) -- cycle;
\fill[fill=color4] (3.2, -1.6) -- (2.4, -0.8) -- (3.2, 0.0) -- cycle;
\fill[fill=color3] (1.6, 0.0) -- (2.4, -0.8) -- (3.2, 0.0) -- cycle;
\fill[fill=color1] (1.6, -1.6) -- (2.4, -0.8) -- (3.2, -1.6) -- cycle;
\node [anchor=north,fill=white] at (2.4, -1.6) {\large $$};
\node () at (2.4, -0.8) [rectangle,draw,fill=none,thick,minimum width=1.6cm, minimum height=1.6cm] {};
\draw [color=black] (1.6, -1.6) -- (3.2, 0.0);
\draw [color=black] (1.6, 0.0) -- (3.2, -1.6);
\node [color=black] at (1.952, -0.8) {$2$};
\node [color=black] at (2.848, -0.8) {$4$};
\node [color=black] at (2.4, -1.312) {$1$};
\node [color=black] at (2.4, -0.288) {$3$};
\fill[fill=color4] (3.2, -1.6) -- (4.0, -0.8) -- (3.2, 0.0) -- cycle;
\fill[fill=color5] (4.8, -1.6) -- (4.0, -0.8) -- (4.8, 0.0) -- cycle;
\fill[fill=color1] (3.2, 0.0) -- (4.0, -0.8) -- (4.8, 0.0) -- cycle;
\fill[fill=color1] (3.2, -1.6) -- (4.0, -0.8) -- (4.8, -1.6) -- cycle;
\node [anchor=north,fill=white] at (4.0, -1.6) {\large $$};
\node () at (4.0, -0.8) [rectangle,draw,fill=none,thick,minimum width=1.6cm, minimum height=1.6cm] {};
\draw [color=black] (3.2, -1.6) -- (4.8, 0.0);
\draw [color=black] (3.2, 0.0) -- (4.8, -1.6);
\node [color=black] at (3.552, -0.8) {$4$};
\node [color=black] at (4.448, -0.8) {$5$};
\node [color=black] at (4.0, -1.312) {$1$};
\node [color=black] at (4.0, -0.288) {$1$};
\fill[fill=color2] (4.8, -1.6) -- (5.6, -0.8) -- (4.8, 0.0) -- cycle;
\fill[fill=color5] (6.4, -1.6) -- (5.6, -0.8) -- (6.4, 0.0) -- cycle;
\fill[fill=color3] (4.8, 0.0) -- (5.6, -0.8) -- (6.4, 0.0) -- cycle;
\fill[fill=color6] (4.8, -1.6) -- (5.6, -0.8) -- (6.4, -1.6) -- cycle;
\node [anchor=north,fill=white] at (5.6, -1.6) {\large $$};
\node () at (5.6, -0.8) [rectangle,draw,fill=none,thick,minimum width=1.6cm, minimum height=1.6cm] {};
\draw [color=black] (4.8, -1.6) -- (6.4, 0.0);
\draw [color=black] (4.8, 0.0) -- (6.4, -1.6);
\node [color=black] at (5.152, -0.8) {$2$};
\node [color=black] at (6.048, -0.8) {$5$};
\node [color=black] at (5.6, -1.312) {$6$};
\node [color=black] at (5.6, -0.288) {$3$};
\fill[fill=color6] (1.6, -3.2) -- (2.4, -2.4) -- (1.6, -1.6) -- cycle;
\fill[fill=color2] (3.2, -3.2) -- (2.4, -2.4) -- (3.2, -1.6) -- cycle;
\fill[fill=color1] (1.6, -1.6) -- (2.4, -2.4) -- (3.2, -1.6) -- cycle;
\fill[fill=color3] (1.6, -3.2) -- (2.4, -2.4) -- (3.2, -3.2) -- cycle;
\node [anchor=north,fill=white] at (2.4, -3.2) {\large $$};
\node () at (2.4, -2.4) [rectangle,draw,fill=none,thick,minimum width=1.6cm, minimum height=1.6cm] {};
\draw [color=black] (1.6, -3.2) -- (3.2, -1.6);
\draw [color=black] (1.6, -1.6) -- (3.2, -3.2);
\node [color=black] at (1.952, -2.4) {$6$};
\node [color=black] at (2.848, -2.4) {$2$};
\node [color=black] at (2.4, -2.912) {$3$};
\node [color=black] at (2.4, -1.888) {$1$};
\fill[fill=color2] (3.2, -3.2) -- (4.0, -2.4) -- (3.2, -1.6) -- cycle;
\fill[fill=color4] (4.8, -3.2) -- (4.0, -2.4) -- (4.8, -1.6) -- cycle;
\fill[fill=color1] (3.2, -1.6) -- (4.0, -2.4) -- (4.8, -1.6) -- cycle;
\fill[fill=color5] (3.2, -3.2) -- (4.0, -2.4) -- (4.8, -3.2) -- cycle;
\node [anchor=north,fill=white] at (4.0, -3.2) {\large $$};
\node () at (4.0, -2.4) [rectangle,draw,fill=none,thick,minimum width=1.6cm, minimum height=1.6cm] {};
\draw [color=black] (3.2, -3.2) -- (4.8, -1.6);
\draw [color=black] (3.2, -1.6) -- (4.8, -3.2);
\node [color=black] at (3.552, -2.4) {$2$};
\node [color=black] at (4.448, -2.4) {$4$};
\node [color=black] at (4.0, -2.912) {$5$};
\node [color=black] at (4.0, -1.888) {$1$};
\fill[fill=color4] (4.8, -3.2) -- (5.6, -2.4) -- (4.8, -1.6) -- cycle;
\fill[fill=color4] (6.4, -3.2) -- (5.6, -2.4) -- (6.4, -1.6) -- cycle;
\fill[fill=color6] (4.8, -1.6) -- (5.6, -2.4) -- (6.4, -1.6) -- cycle;
\fill[fill=color6] (4.8, -3.2) -- (5.6, -2.4) -- (6.4, -3.2) -- cycle;
\node [anchor=north,fill=white] at (5.6, -3.2) {\large $$};
\node () at (5.6, -2.4) [rectangle,draw,fill=none,thick,minimum width=1.6cm, minimum height=1.6cm] {};
\draw [color=black] (4.8, -3.2) -- (6.4, -1.6);
\draw [color=black] (4.8, -1.6) -- (6.4, -3.2);
\node [color=black] at (5.152, -2.4) {$4$};
\node [color=black] at (6.048, -2.4) {$4$};
\node [color=black] at (5.6, -2.912) {$6$};
\node [color=black] at (5.6, -1.888) {$6$};
\draw [line width = 5pt, color = black] (4.8,-1.6) -- (4.8,0.0);
\draw [line width = 2pt, color = red] (4.8,-1.6) -- (4.8,0.0);
\end{tikzpicture}

%% file: graph2.tikz
\begin{tikzpicture}

\matrix[nodes={draw}, row sep=5mm, column sep=5mm] {
 \node[circle] (V1) {$v_1$}; &  & \node[circle] (V2) {$v_2$}; & \\
 & \node[circle] (V3) {$v_3$}; \\
 \node[circle] (V4) {$v_4$}; &  & \node[circle] (V5) {$v_5$}; & \\
};

\path (V1) edge[-latex] (V2);
\path (V1) edge[-latex,ultra thick] (V4);
\path (V2) edge[-latex,ultra thick] (V5);
\path (V3) edge[-latex,ultra thick] (V2);
\path (V3) edge[-latex] (V5);
\path (V4) edge[-latex,ultra thick] (V3);

\end{tikzpicture}

%% file: graph-overlaid2.tikz
\begin{tikzpicture}
\node [anchor=north,fill=white] at (0.8, 8.32) {\large $\mathcal{T}_{v_1}$};
\node (V1) at (0.8, 9.12) [rectangle,draw,fill=gray!30,thick,minimum width=1.6cm, minimum height=1.6cm] {};
\draw [color=black] (0.0, 8.32) -- (1.6, 9.92);
\draw [color=black] (0.0, 9.92) -- (1.6, 8.32);
\node [color=black] at (0.352, 9.12) {$\scriptscriptstyle +I_{v_1}$};
\node [color=black] at (1.248, 9.12) {$\scriptscriptstyle +O_{v_1}$};
\node [color=black] at (0.8, 8.608) {$\scriptscriptstyle +U_{v_1}$};
\node [color=black] at (0.8, 9.632) {$\scriptscriptstyle +U_{v_1}$};
\node [anchor=north,fill=white] at (0.8, 4.16) {\large $\mathcal{T}_{e_{14}}$};
\node (E14) at (0.8, 4.96) [rectangle,draw,fill=white,thick,minimum width=1.6cm, minimum height=1.6cm] {};
\draw [color=black] (0.0, 4.16) -- (1.6, 5.76);
\draw [color=black] (0.0, 5.76) -- (1.6, 4.16);
\node [color=black] at (0.352, 4.96) {$\scriptscriptstyle -O_{v_1}$};
\node [color=black] at (1.248, 4.96) {$\scriptscriptstyle -I_{v_4}$};
\node [color=black] at (0.8, 4.448) {$\scriptscriptstyle -X$};
\node [color=black] at (0.8, 5.472) {$\scriptscriptstyle +X$};
\node [anchor=north,fill=white] at (0.8, 0.0) {\large $\mathcal{T}_{v_4}$};
\node (V4) at (0.8, 0.8) [rectangle,draw,fill=gray!30,thick,minimum width=1.6cm, minimum height=1.6cm] {};
\draw [color=black] (0.0, 0.0) -- (1.6, 1.6);
\draw [color=black] (0.0, 1.6) -- (1.6, 0.0);
\node [color=black] at (0.352, 0.8) {$\scriptscriptstyle +I_{v_4}$};
\node [color=black] at (1.248, 0.8) {$\scriptscriptstyle +O_{v_4}$};
\node [color=black] at (0.8, 0.288) {$\scriptscriptstyle +U_{v_4}$};
\node [color=black] at (0.8, 1.312) {$\scriptscriptstyle +U_{v_4}$};
\node [anchor=north,fill=white] at (2.88, 2.08) {\large $\mathcal{T}_{e_{43}}$};
\node (E43) at (2.88, 2.88) [rectangle,draw,fill=white,thick,minimum width=1.6cm, minimum height=1.6cm] {};
\draw [color=black] (2.08, 2.08) -- (3.68, 3.68);
\draw [color=black] (2.08, 3.68) -- (3.68, 2.08);
\node [color=black] at (2.432, 2.88) {$\scriptscriptstyle -O_{v_4}$};
\node [color=black] at (3.328, 2.88) {$\scriptscriptstyle -I_{v_3}$};
\node [color=black] at (2.88, 2.368) {$\scriptscriptstyle -X$};
\node [color=black] at (2.88, 3.392) {$\scriptscriptstyle +X$};
\node [anchor=north,fill=white] at (4.96, 4.16) {\large $\mathcal{T}_{v_3}$};
\node (V3) at (4.96, 4.96) [rectangle,draw,fill=gray!30,thick,minimum width=1.6cm, minimum height=1.6cm] {};
\draw [color=black] (4.16, 4.16) -- (5.76, 5.76);
\draw [color=black] (4.16, 5.76) -- (5.76, 4.16);
\node [color=black] at (4.512, 4.96) {$\scriptscriptstyle +I_{v_3}$};
\node [color=black] at (5.408, 4.96) {$\scriptscriptstyle +O_{v_3}$};
\node [color=black] at (4.96, 4.448) {$\scriptscriptstyle +U_{v_3}$};
\node [color=black] at (4.96, 5.472) {$\scriptscriptstyle +U_{v_3}$};
\node [anchor=north,fill=white] at (4.96, 8.32) {\large $\mathcal{T}_{e_{12}}$};
\node (E12) at (4.96, 9.12) [rectangle,draw,fill=white,thick,minimum width=1.6cm, minimum height=1.6cm] {};
\draw [color=black] (4.16, 8.32) -- (5.76, 9.92);
\draw [color=black] (4.16, 9.92) -- (5.76, 8.32);
\node [color=black] at (4.512, 9.12) {$\scriptscriptstyle -O_{v_1}$};
\node [color=black] at (5.408, 9.12) {$\scriptscriptstyle -I_{v_2}$};
\node [color=black] at (4.96, 8.608) {$\scriptscriptstyle -X$};
\node [color=black] at (4.96, 9.632) {$\scriptscriptstyle +X$};
\node [anchor=north,fill=white] at (7.04, 6.24) {\large $\mathcal{T}_{e_{32}}$};
\node (E32) at (7.04, 7.04) [rectangle,draw,fill=white,thick,minimum width=1.6cm, minimum height=1.6cm] {};
\draw [color=black] (6.24, 6.24) -- (7.84, 7.84);
\draw [color=black] (6.24, 7.84) -- (7.84, 6.24);
\node [color=black] at (6.592, 7.04) {$\scriptscriptstyle -O_{v_3}$};
\node [color=black] at (7.488, 7.04) {$\scriptscriptstyle -I_{v_2}$};
\node [color=black] at (7.04, 6.528) {$\scriptscriptstyle -X$};
\node [color=black] at (7.04, 7.552) {$\scriptscriptstyle +X$};
\node [anchor=north,fill=white] at (7.04, 2.08) {\large $\mathcal{T}_{e_{35}}$};
\node (E35) at (7.04, 2.88) [rectangle,draw,fill=white,thick,minimum width=1.6cm, minimum height=1.6cm] {};
\draw [color=black] (6.24, 2.08) -- (7.84, 3.68);
\draw [color=black] (6.24, 3.68) -- (7.84, 2.08);
\node [color=black] at (6.592, 2.88) {$\scriptscriptstyle -O_{v_3}$};
\node [color=black] at (7.488, 2.88) {$\scriptscriptstyle -I_{v_5}$};
\node [color=black] at (7.04, 2.368) {$\scriptscriptstyle -X$};
\node [color=black] at (7.04, 3.392) {$\scriptscriptstyle +X$};
\node [anchor=north,fill=white] at (9.12, 8.32) {\large $\mathcal{T}_{v_2}$};
\node (V2) at (9.12, 9.12) [rectangle,draw,fill=gray!30,thick,minimum width=1.6cm, minimum height=1.6cm] {};
\draw [color=black] (8.32, 8.32) -- (9.92, 9.92);
\draw [color=black] (8.32, 9.92) -- (9.92, 8.32);
\node [color=black] at (8.672, 9.12) {$\scriptscriptstyle +I_{v_2}$};
\node [color=black] at (9.568, 9.12) {$\scriptscriptstyle +O_{v_2}$};
\node [color=black] at (9.12, 8.608) {$\scriptscriptstyle +U_{v_2}$};
\node [color=black] at (9.12, 9.632) {$\scriptscriptstyle +U_{v_2}$};
\node [anchor=north,fill=white] at (9.12, 4.16) {\large $\mathcal{T}_{e_{25}}$};
\node (E25) at (9.12, 4.96) [rectangle,draw,fill=white,thick,minimum width=1.6cm, minimum height=1.6cm] {};
\draw [color=black] (8.32, 4.16) -- (9.92, 5.76);
\draw [color=black] (8.32, 5.76) -- (9.92, 4.16);
\node [color=black] at (8.672, 4.96) {$\scriptscriptstyle -O_{v_2}$};
\node [color=black] at (9.568, 4.96) {$\scriptscriptstyle -I_{v_5}$};
\node [color=black] at (9.12, 4.448) {$\scriptscriptstyle -X$};
\node [color=black] at (9.12, 5.472) {$\scriptscriptstyle +X$};
\node [anchor=north,fill=white] at (9.12, 0.0) {\large $\mathcal{T}_{v_5}$};
\node (V5) at (9.12, 0.8) [rectangle,draw,fill=gray!30,thick,minimum width=1.6cm, minimum height=1.6cm] {};
\draw [color=black] (8.32, 0.0) -- (9.92, 1.6);
\draw [color=black] (8.32, 1.6) -- (9.92, 0.0);
\node [color=black] at (8.672, 0.8) {$\scriptscriptstyle +I_{v_5}$};
\node [color=black] at (9.568, 0.8) {$\scriptscriptstyle +O_{v_5}$};
\node [color=black] at (9.12, 0.288) {$\scriptscriptstyle +U_{v_5}$};
\node [color=black] at (9.12, 1.312) {$\scriptscriptstyle +U_{v_5}$};
\node [anchor=north,fill=white] at (12.24, 4.16) {\large $\mathcal{T}_B$};
\node (TB) at (12.24, 4.96) [rectangle,draw,fill=gray!60,thick,minimum width=1.6cm, minimum height=1.6cm] {};
\draw [color=black] (11.44, 4.16) -- (13.04, 5.76);
\draw [color=black] (11.44, 5.76) -- (13.04, 4.16);
\node [color=black] at (11.792, 4.96) {$\scriptscriptstyle -O_{v_5}$};
\node [color=black] at (12.688, 4.96) {$\scriptscriptstyle -X$};
\node [color=black] at (12.24, 4.448) {$\scriptscriptstyle +U_B$};
\node [color=black] at (12.24, 5.472) {$\scriptscriptstyle +U_B$};
\begin{scope}[on background layer]
\path (V1) edge[-latex,thick,shorten <=5pt,shorten >=5pt] (V2);
\path (V1) edge[-latex,line width=2pt,shorten <=5pt,shorten >=5pt] (V4);
\path (V2) edge[-latex,line width=2pt,shorten <=5pt,shorten >=5pt] (V5);
\path (V3) edge[-latex,line width=2pt,shorten <=5pt,shorten >=5pt] (V2);
\path (V3) edge[-latex,thick,shorten <=5pt,shorten >=5pt] (V5);
\path (V4) edge[-latex,line width=2pt,shorten <=5pt,shorten >=5pt] (V3);
\end{scope}
\end{tikzpicture}

%% file: exampletiling2_paper.tikz
\begin{tikzpicture}
\draw (0.0,0.0) rectangle (16.0,1.6);
\node [anchor=north,fill=white] at (0.8, 0.0) {\large $\mathcal{T}_{v_1}$};
\node (V1) at (0.8, 0.8) [rectangle,draw,fill=gray!30,thick,minimum width=1.6cm, minimum height=1.6cm] {};
\draw [color=black] (0.0, 0.0) -- (1.6, 1.6);
\draw [color=black] (0.0, 1.6) -- (1.6, 0.0);
\node [color=black] at (0.352, 0.8) {$\scriptscriptstyle +I_{v_1}$};
\node [color=black] at (1.248, 0.8) {$\scriptscriptstyle +O_{v_1}$};
\node [color=black] at (0.8, 0.288) {$\scriptscriptstyle +U_{v_1}$};
\node [color=black] at (0.8, 1.312) {$\scriptscriptstyle +U_{v_1}$};
\node [anchor=north,fill=white] at (2.4, 0.0) {\large $\mathcal{T}_{e_{14}}$};
\node (E14) at (2.4, 0.8) [rectangle,draw,fill=white,thick,minimum width=1.6cm, minimum height=1.6cm] {};
\draw [color=black] (1.6, 0.0) -- (3.2, 1.6);
\draw [color=black] (1.6, 1.6) -- (3.2, 0.0);
\node [color=black] at (1.952, 0.8) {$\scriptscriptstyle -O_{v_1}$};
\node [color=black] at (2.848, 0.8) {$\scriptscriptstyle -I_{v_4}$};
\node [color=black] at (2.4, 0.288) {$\scriptscriptstyle -X$};
\node [color=black] at (2.4, 1.312) {$\scriptscriptstyle +X$};
\node [anchor=north,fill=white] at (4.0, 0.0) {\large $\mathcal{T}_{v_4}$};
\node (V4) at (4.0, 0.8) [rectangle,draw,fill=gray!30,thick,minimum width=1.6cm, minimum height=1.6cm] {};
\draw [color=black] (3.2, 0.0) -- (4.8, 1.6);
\draw [color=black] (3.2, 1.6) -- (4.8, 0.0);
\node [color=black] at (3.552, 0.8) {$\scriptscriptstyle +I_{v_4}$};
\node [color=black] at (4.448, 0.8) {$\scriptscriptstyle +O_{v_4}$};
\node [color=black] at (4.0, 0.288) {$\scriptscriptstyle +U_{v_4}$};
\node [color=black] at (4.0, 1.312) {$\scriptscriptstyle +U_{v_4}$};
\node [anchor=north,fill=white] at (5.6, 0.0) {\large $\mathcal{T}_{e_{43}}$};
\node (E43) at (5.6, 0.8) [rectangle,draw,fill=white,thick,minimum width=1.6cm, minimum height=1.6cm] {};
\draw [color=black] (4.8, 0.0) -- (6.4, 1.6);
\draw [color=black] (4.8, 1.6) -- (6.4, 0.0);
\node [color=black] at (5.152, 0.8) {$\scriptscriptstyle -O_{v_4}$};
\node [color=black] at (6.048, 0.8) {$\scriptscriptstyle -I_{v_3}$};
\node [color=black] at (5.6, 0.288) {$\scriptscriptstyle -X$};
\node [color=black] at (5.6, 1.312) {$\scriptscriptstyle +X$};
\node [anchor=north,fill=white] at (7.2, 0.0) {\large $\mathcal{T}_{v_3}$};
\node (V3) at (7.2, 0.8) [rectangle,draw,fill=gray!30,thick,minimum width=1.6cm, minimum height=1.6cm] {};
\draw [color=black] (6.4, 0.0) -- (8.0, 1.6);
\draw [color=black] (6.4, 1.6) -- (8.0, 0.0);
\node [color=black] at (6.752, 0.8) {$\scriptscriptstyle +I_{v_3}$};
\node [color=black] at (7.648, 0.8) {$\scriptscriptstyle +O_{v_3}$};
\node [color=black] at (7.2, 0.288) {$\scriptscriptstyle +U_{v_3}$};
\node [color=black] at (7.2, 1.312) {$\scriptscriptstyle +U_{v_3}$};
\node [anchor=north,fill=white] at (8.8, 0.0) {\large $\mathcal{T}_{e_{32}}$};
\node (E32) at (8.8, 0.8) [rectangle,draw,fill=white,thick,minimum width=1.6cm, minimum height=1.6cm] {};
\draw [color=black] (8.0, 0.0) -- (9.6, 1.6);
\draw [color=black] (8.0, 1.6) -- (9.6, 0.0);
\node [color=black] at (8.352, 0.8) {$\scriptscriptstyle -O_{v_3}$};
\node [color=black] at (9.248, 0.8) {$\scriptscriptstyle -I_{v_2}$};
\node [color=black] at (8.8, 0.288) {$\scriptscriptstyle -X$};
\node [color=black] at (8.8, 1.312) {$\scriptscriptstyle +X$};
\node [anchor=north,fill=white] at (10.4, 0.0) {\large $\mathcal{T}_{v_2}$};
\node (V2) at (10.4, 0.8) [rectangle,draw,fill=gray!30,thick,minimum width=1.6cm, minimum height=1.6cm] {};
\draw [color=black] (9.6, 0.0) -- (11.2, 1.6);
\draw [color=black] (9.6, 1.6) -- (11.2, 0.0);
\node [color=black] at (9.952, 0.8) {$\scriptscriptstyle +I_{v_2}$};
\node [color=black] at (10.848, 0.8) {$\scriptscriptstyle +O_{v_2}$};
\node [color=black] at (10.4, 0.288) {$\scriptscriptstyle +U_{v_2}$};
\node [color=black] at (10.4, 1.312) {$\scriptscriptstyle +U_{v_2}$};
\node [anchor=north,fill=white] at (12.0, 0.0) {\large $\mathcal{T}_{e_{25}}$};
\node (E25) at (12.0, 0.8) [rectangle,draw,fill=white,thick,minimum width=1.6cm, minimum height=1.6cm] {};
\draw [color=black] (11.2, 0.0) -- (12.8, 1.6);
\draw [color=black] (11.2, 1.6) -- (12.8, 0.0);
\node [color=black] at (11.552, 0.8) {$\scriptscriptstyle -O_{v_2}$};
\node [color=black] at (12.448, 0.8) {$\scriptscriptstyle -I_{v_5}$};
\node [color=black] at (12.0, 0.288) {$\scriptscriptstyle -X$};
\node [color=black] at (12.0, 1.312) {$\scriptscriptstyle +X$};
\node [anchor=north,fill=white] at (13.6, 0.0) {\large $\mathcal{T}_{v_5}$};
\node (V5) at (13.6, 0.8) [rectangle,draw,fill=gray!30,thick,minimum width=1.6cm, minimum height=1.6cm] {};
\draw [color=black] (12.8, 0.0) -- (14.4, 1.6);
\draw [color=black] (12.8, 1.6) -- (14.4, 0.0);
\node [color=black] at (13.152, 0.8) {$\scriptscriptstyle +I_{v_5}$};
\node [color=black] at (14.048, 0.8) {$\scriptscriptstyle +O_{v_5}$};
\node [color=black] at (13.6, 0.288) {$\scriptscriptstyle +U_{v_5}$};
\node [color=black] at (13.6, 1.312) {$\scriptscriptstyle +U_{v_5}$};
\node [anchor=north,fill=white] at (15.2, 0.0) {\large $\mathcal{T}_B$};
\node (TB) at (15.2, 0.8) [rectangle,draw,fill=gray!60,thick,minimum width=1.6cm, minimum height=1.6cm] {};
\draw [color=black] (14.4, 0.0) -- (16.0, 1.6);
\draw [color=black] (14.4, 1.6) -- (16.0, 0.0);
\node [color=black] at (14.752, 0.8) {$\scriptscriptstyle -O_{v_5}$};
\node [color=black] at (15.648, 0.8) {$\scriptscriptstyle -X$};
\node [color=black] at (15.2, 0.288) {$\scriptscriptstyle +U_B$};
\node [color=black] at (15.2, 1.312) {$\scriptscriptstyle +U_B$};
\node [anchor=north,fill=white] at (16.8, 0.0) {\large $\mathcal{T}_{e_{12}}$};
\node (E12) at (16.8, 0.8) [rectangle,draw,fill=white,thick,minimum width=1.6cm, minimum height=1.6cm] {};
\draw [color=black] (16.0, 0.0) -- (17.6, 1.6);
\draw [color=black] (16.0, 1.6) -- (17.6, 0.0);
\node [color=black] at (16.352, 0.8) {$\scriptscriptstyle +X$};
\node [color=black] at (17.248, 0.8) {$\scriptscriptstyle -X$};
\node [color=black] at (16.8, 0.288) {$\scriptscriptstyle -O_{v_1}$};
\node [color=black] at (16.8, 1.312) {$\scriptscriptstyle -I_{v_2}$};
\node [anchor=north,fill=white] at (18.4, 0.0) {\large $\mathcal{T}_{e_{35}}$};
\node (E35) at (18.4, 0.8) [rectangle,draw,fill=white,thick,minimum width=1.6cm, minimum height=1.6cm] {};
\draw [color=black] (17.6, 0.0) -- (19.2, 1.6);
\draw [color=black] (17.6, 1.6) -- (19.2, 0.0);
\node [color=black] at (17.952, 0.8) {$\scriptscriptstyle +X$};
\node [color=black] at (18.848, 0.8) {$\scriptscriptstyle -X$};
\node [color=black] at (18.4, 0.288) {$\scriptscriptstyle -O_{v_3}$};
\node [color=black] at (18.4, 1.312) {$\scriptscriptstyle -I_{v_5}$};
\draw[decoration={brace,raise=5pt},decorate] (14.4,1.6) -- node[above=6pt] {\large bridge tile} (16.0,1.6);
\draw[decoration={brace,mirror,raise=18pt},decorate] (0.0,0.0) -- node[below=20pt] {\large Hamiltonian path} (14.4,0.0);
\draw[decoration={brace,mirror,raise=18pt},decorate] (16.0,0.0) -- node[below=20pt] {\large unused edges} (19.2,0.0);
\end{tikzpicture}

%% file: othervertex_paper.tikz
\begin{tikzpicture}
\node [anchor=north,fill=white] at (0.8, 0.0) {\large $\mathcal{T}_{v_1}$};
\node (V1) at (0.8, 0.8) [rectangle,draw,fill=gray!30,thick,minimum width=1.6cm, minimum height=1.6cm] {};
\draw [color=black] (0.0, 0.0) -- (1.6, 1.6);
\draw [color=black] (0.0, 1.6) -- (1.6, 0.0);
\node [color=black] at (0.352, 0.8) {$\scriptstyle I_{v_1}$};
\node [color=black] at (1.248, 0.8) {$\scriptstyle O_{v_1}$};
\node [color=black] at (0.8, 0.288) {$\scriptstyle U_{v_1}$};
\node [color=black] at (0.8, 1.312) {$\scriptstyle U_{v_1}$};
\node [anchor=north,fill=white] at (3.2, 0.0) {\large $\mathcal{T}_{v_2}$};
\node (V2) at (3.2, 0.8) [rectangle,draw,fill=gray!30,thick,minimum width=1.6cm, minimum height=1.6cm] {};
\draw [color=black] (2.4, 0.0) -- (4.0, 1.6);
\draw [color=black] (2.4, 1.6) -- (4.0, 0.0);
\node [color=black] at (2.752, 0.8) {$\scriptstyle I_{v_2}$};
\node [color=black] at (3.648, 0.8) {$\scriptstyle O_{v_2}$};
\node [color=black] at (3.2, 0.288) {$\scriptstyle U_{v_2}$};
\node [color=black] at (3.2, 1.312) {$\scriptstyle U_{v_2}$};
\node [anchor=north,fill=white] at (5.6, 0.0) {\large $\mathcal{T}_{v_3}$};
\node (V3) at (5.6, 0.8) [rectangle,draw,fill=gray!30,thick,minimum width=1.6cm, minimum height=1.6cm] {};
\draw [color=black] (4.8, 0.0) -- (6.4, 1.6);
\draw [color=black] (4.8, 1.6) -- (6.4, 0.0);
\node [color=black] at (5.152, 0.8) {$\scriptstyle I_{v_3}$};
\node [color=black] at (6.048, 0.8) {$\scriptstyle O_{v_3}$};
\node [color=black] at (5.6, 0.288) {$\scriptstyle U_{v_3}$};
\node [color=black] at (5.6, 1.312) {$\scriptstyle U_{v_3}$};
\node [anchor=north,fill=white] at (8.0, 0.0) {\large $\mathcal{T}_{v_4}$};
\node (V4) at (8.0, 0.8) [rectangle,draw,fill=gray!30,thick,minimum width=1.6cm, minimum height=1.6cm] {};
\draw [color=black] (7.2, 0.0) -- (8.8, 1.6);
\draw [color=black] (7.2, 1.6) -- (8.8, 0.0);
\node [color=black] at (7.552, 0.8) {$\scriptstyle I_{v_4}$};
\node [color=black] at (8.448, 0.8) {$\scriptstyle O_{v_4}$};
\node [color=black] at (8.0, 0.288) {$\scriptstyle U_{v_4}$};
\node [color=black] at (8.0, 1.312) {$\scriptstyle U_{v_4}$};
\node [anchor=north,fill=white] at (10.4, 0.0) {\large $\mathcal{T}_{v_5}$};
\node (V5) at (10.4, 0.8) [rectangle,draw,fill=gray!30,thick,minimum width=1.6cm, minimum height=1.6cm] {};
\draw [color=black] (9.6, 0.0) -- (11.2, 1.6);
\draw [color=black] (9.6, 1.6) -- (11.2, 0.0);
\node [color=black] at (9.952, 0.8) {$\scriptstyle I_{v_5}$};
\node [color=black] at (10.848, 0.8) {$\scriptstyle O_{v_5}$};
\node [color=black] at (10.4, 0.288) {$\scriptstyle U_{v_5}$};
\node [color=black] at (10.4, 1.312) {$\scriptstyle U_{v_5}$};
\node [anchor=north,fill=white] at (12.8, 0.0) {\large $\mathcal{T}_B$};
\node (TB) at (12.8, 0.8) [rectangle,draw,fill=gray!60,thick,minimum width=1.6cm, minimum height=1.6cm] {};
\draw [color=black] (12.0, 0.0) -- (13.6, 1.6);
\draw [color=black] (12.0, 1.6) -- (13.6, 0.0);
\node [color=black] at (12.352, 0.8) {$\scriptstyle O_{v_5}$};
\node [color=black] at (13.248, 0.8) {$\scriptstyle X$};
\node [color=black] at (12.8, 0.288) {$\scriptstyle U_B$};
\node [color=black] at (12.8, 1.312) {$\scriptstyle U_B$};
\end{tikzpicture}

%% file: edges_paper.tikz
\begin{tikzpicture}
\node [anchor=north,fill=white] at (0.8, 0.0) {\large $\mathcal{T}_{e_{12}}$};
\node (E12) at (0.8, 0.8) [rectangle,draw,fill=white,thick,minimum width=1.6cm, minimum height=1.6cm] {};
\draw [color=black] (0.0, 0.0) -- (1.6, 1.6);
\draw [color=black] (0.0, 1.6) -- (1.6, 0.0);
\node [color=black] at (0.352, 0.8) {$\scriptstyle O_{v_1}$};
\node [color=black] at (1.248, 0.8) {$\scriptstyle I_{v_2}$};
\node [color=black] at (0.8, 0.288) {$\scriptstyle X$};
\node [color=black] at (0.8, 1.312) {$\scriptstyle X$};
\node [anchor=north,fill=white] at (3.2, 0.0) {\large $\mathcal{T}_{e_{14}}$};
\node (E14) at (3.2, 0.8) [rectangle,draw,fill=white,thick,minimum width=1.6cm, minimum height=1.6cm] {};
\draw [color=black] (2.4, 0.0) -- (4.0, 1.6);
\draw [color=black] (2.4, 1.6) -- (4.0, 0.0);
\node [color=black] at (2.752, 0.8) {$\scriptstyle O_{v_1}$};
\node [color=black] at (3.648, 0.8) {$\scriptstyle I_{v_4}$};
\node [color=black] at (3.2, 0.288) {$\scriptstyle X$};
\node [color=black] at (3.2, 1.312) {$\scriptstyle X$};
\node [anchor=north,fill=white] at (5.6, 0.0) {\large $\mathcal{T}_{e_{25}}$};
\node (E25) at (5.6, 0.8) [rectangle,draw,fill=white,thick,minimum width=1.6cm, minimum height=1.6cm] {};
\draw [color=black] (4.8, 0.0) -- (6.4, 1.6);
\draw [color=black] (4.8, 1.6) -- (6.4, 0.0);
\node [color=black] at (5.152, 0.8) {$\scriptstyle O_{v_2}$};
\node [color=black] at (6.048, 0.8) {$\scriptstyle I_{v_5}$};
\node [color=black] at (5.6, 0.288) {$\scriptstyle X$};
\node [color=black] at (5.6, 1.312) {$\scriptstyle X$};
\node [anchor=north,fill=white] at (8.0, 0.0) {\large $\mathcal{T}_{e_{32}}$};
\node (E32) at (8.0, 0.8) [rectangle,draw,fill=white,thick,minimum width=1.6cm, minimum height=1.6cm] {};
\draw [color=black] (7.2, 0.0) -- (8.8, 1.6);
\draw [color=black] (7.2, 1.6) -- (8.8, 0.0);
\node [color=black] at (7.552, 0.8) {$\scriptstyle O_{v_3}$};
\node [color=black] at (8.448, 0.8) {$\scriptstyle I_{v_2}$};
\node [color=black] at (8.0, 0.288) {$\scriptstyle X$};
\node [color=black] at (8.0, 1.312) {$\scriptstyle X$};
\node [anchor=north,fill=white] at (10.4, 0.0) {\large $\mathcal{T}_{e_{35}}$};
\node (E35) at (10.4, 0.8) [rectangle,draw,fill=white,thick,minimum width=1.6cm, minimum height=1.6cm] {};
\draw [color=black] (9.6, 0.0) -- (11.2, 1.6);
\draw [color=black] (9.6, 1.6) -- (11.2, 0.0);
\node [color=black] at (9.952, 0.8) {$\scriptstyle O_{v_3}$};
\node [color=black] at (10.848, 0.8) {$\scriptstyle I_{v_5}$};
\node [color=black] at (10.4, 0.288) {$\scriptstyle X$};
\node [color=black] at (10.4, 1.312) {$\scriptstyle X$};
\node [anchor=north,fill=white] at (12.8, 0.0) {\large $\mathcal{T}_{e_{43}}$};
\node (E43) at (12.8, 0.8) [rectangle,draw,fill=white,thick,minimum width=1.6cm, minimum height=1.6cm] {};
\draw [color=black] (12.0, 0.0) -- (13.6, 1.6);
\draw [color=black] (12.0, 1.6) -- (13.6, 0.0);
\node [color=black] at (12.352, 0.8) {$\scriptstyle O_{v_4}$};
\node [color=black] at (13.248, 0.8) {$\scriptstyle I_{v_3}$};
\node [color=black] at (12.8, 0.288) {$\scriptstyle X$};
\node [color=black] at (12.8, 1.312) {$\scriptstyle X$};
\end{tikzpicture}

%% file: exampletiling-algo_paper.tikz
\begin{tikzpicture}
\draw (0.0,0.0) rectangle (16.0,1.6);
\node (blank) at (0.8, 0.8) [draw,fill=white,thick,minimum width=1.6cm, minimum height=1.6cm] {};
\node [anchor=north,fill=white] at (2.4, 0.0) {\large $\mathcal{T}_{e_{12}}$};
\node (E12) at (2.4, 0.8) [rectangle,draw,fill=white,thick,minimum width=1.6cm, minimum height=1.6cm] {};
\draw [color=black] (1.6, 0.0) -- (3.2, 1.6);
\draw [color=black] (1.6, 1.6) -- (3.2, 0.0);
\node [color=black] at (1.952, 0.8) {$\scriptstyle O_{v_1}$};
\node [color=black] at (2.848, 0.8) {$\scriptstyle I_{v_2}$};
\node [color=black] at (2.4, 0.288) {$\scriptstyle X$};
\node [color=black] at (2.4, 1.312) {$\scriptstyle X$};
\node (blank) at (4.0, 0.8) [draw,fill=white,thick,minimum width=1.6cm, minimum height=1.6cm] {};
\node [anchor=north,fill=white] at (5.6, 0.0) {\large $\mathcal{T}_{v_1}$};
\node (V1) at (5.6, 0.8) [rectangle,draw,fill=gray!30,thick,minimum width=1.6cm, minimum height=1.6cm] {};
\draw [color=black] (4.8, 0.0) -- (6.4, 1.6);
\draw [color=black] (4.8, 1.6) -- (6.4, 0.0);
\node [color=black] at (5.152, 0.8) {$\scriptstyle I_{v_1}$};
\node [color=black] at (6.048, 0.8) {$\scriptstyle O_{v_1}$};
\node [color=black] at (5.6, 0.288) {$\scriptstyle U_{v_1}$};
\node [color=black] at (5.6, 1.312) {$\scriptstyle U_{v_1}$};
\node [anchor=north,fill=white] at (7.2, 0.0) {\large $\mathcal{T}_{e_{14}}$};
\node (E14) at (7.2, 0.8) [rectangle,draw,fill=white,thick,minimum width=1.6cm, minimum height=1.6cm] {};
\draw [color=black] (6.4, 0.0) -- (8.0, 1.6);
\draw [color=black] (6.4, 1.6) -- (8.0, 0.0);
\node [color=black] at (6.752, 0.8) {$\scriptstyle O_{v_1}$};
\node [color=black] at (7.648, 0.8) {$\scriptstyle I_{v_4}$};
\node [color=black] at (7.2, 0.288) {$\scriptstyle X$};
\node [color=black] at (7.2, 1.312) {$\scriptstyle X$};
\node [anchor=north,fill=white] at (8.8, 0.0) {\large $\mathcal{T}_{v_4}$};
\node (V4) at (8.8, 0.8) [rectangle,draw,fill=gray!30,thick,minimum width=1.6cm, minimum height=1.6cm] {};
\draw [color=black] (8.0, 0.0) -- (9.6, 1.6);
\draw [color=black] (8.0, 1.6) -- (9.6, 0.0);
\node [color=black] at (8.352, 0.8) {$\scriptstyle I_{v_4}$};
\node [color=black] at (9.248, 0.8) {$\scriptstyle O_{v_4}$};
\node [color=black] at (8.8, 0.288) {$\scriptstyle U_{v_4}$};
\node [color=black] at (8.8, 1.312) {$\scriptstyle U_{v_4}$};
\node [anchor=north,fill=white] at (10.4, 0.0) {\large $\mathcal{T}_{e_{43}}$};
\node (E43) at (10.4, 0.8) [rectangle,draw,fill=white,thick,minimum width=1.6cm, minimum height=1.6cm] {};
\draw [color=black] (9.6, 0.0) -- (11.2, 1.6);
\draw [color=black] (9.6, 1.6) -- (11.2, 0.0);
\node [color=black] at (9.952, 0.8) {$\scriptstyle O_{v_4}$};
\node [color=black] at (10.848, 0.8) {$\scriptstyle I_{v_3}$};
\node [color=black] at (10.4, 0.288) {$\scriptstyle X$};
\node [color=black] at (10.4, 1.312) {$\scriptstyle X$};
\node [anchor=north,fill=white] at (12.0, 0.0) {\large $\mathcal{T}_{v_3}$};
\node (V3) at (12.0, 0.8) [rectangle,draw,fill=gray!30,thick,minimum width=1.6cm, minimum height=1.6cm] {};
\draw [color=black] (11.2, 0.0) -- (12.8, 1.6);
\draw [color=black] (11.2, 1.6) -- (12.8, 0.0);
\node [color=black] at (11.552, 0.8) {$\scriptstyle I_{v_3}$};
\node [color=black] at (12.448, 0.8) {$\scriptstyle O_{v_3}$};
\node [color=black] at (12.0, 0.288) {$\scriptstyle U_{v_3}$};
\node [color=black] at (12.0, 1.312) {$\scriptstyle U_{v_3}$};
\node [anchor=north,fill=white] at (13.6, 0.0) {\large $\mathcal{T}_{e_{35}}$};
\node (E35) at (13.6, 0.8) [rectangle,draw,fill=white,thick,minimum width=1.6cm, minimum height=1.6cm] {};
\draw [color=black] (12.8, 0.0) -- (14.4, 1.6);
\draw [color=black] (12.8, 1.6) -- (14.4, 0.0);
\node [color=black] at (13.152, 0.8) {$\scriptstyle O_{v_3}$};
\node [color=black] at (14.048, 0.8) {$\scriptstyle I_{v_5}$};
\node [color=black] at (13.6, 0.288) {$\scriptstyle X$};
\node [color=black] at (13.6, 1.312) {$\scriptstyle X$};
\node [anchor=north,fill=white] at (15.2, 0.0) {\large $\mathcal{T}_{e_{25}}$};
\node (E52) at (15.2, 0.8) [rectangle,draw,fill=white,thick,minimum width=1.6cm, minimum height=1.6cm] {};
\draw [color=black] (14.4, 0.0) -- (16.0, 1.6);
\draw [color=black] (14.4, 1.6) -- (16.0, 0.0);
\node [color=black] at (14.752, 0.8) {$\scriptstyle I_{v_5}$};
\node [color=black] at (15.648, 0.8) {$\scriptstyle O_{v_2}$};
\node [color=black] at (15.2, 0.288) {$\scriptstyle X$};
\node [color=black] at (15.2, 1.312) {$\scriptstyle X$};
\node [anchor=north,fill=white] at (16.8, 0.0) {\large $\mathcal{T}_{v_2}$};
\node (V2) at (16.8, 0.8) [rectangle,draw,fill=gray!30,thick,minimum width=1.6cm, minimum height=1.6cm] {};
\draw [color=black] (16.0, 0.0) -- (17.6, 1.6);
\draw [color=black] (16.0, 1.6) -- (17.6, 0.0);
\node [color=black] at (16.352, 0.8) {$\scriptstyle O_{v_2}$};
\node [color=black] at (17.248, 0.8) {$\scriptstyle I_{v_2}$};
\node [color=black] at (16.8, 0.288) {$\scriptstyle U_{v_2}$};
\node [color=black] at (16.8, 1.312) {$\scriptstyle U_{v_2}$};
\node [anchor=north,fill=white] at (18.4, 0.0) {\large $\mathcal{T}_{e_{32}}$};
\node (E23) at (18.4, 0.8) [rectangle,draw,fill=white,thick,minimum width=1.6cm, minimum height=1.6cm] {};
\draw [color=black] (17.6, 0.0) -- (19.2, 1.6);
\draw [color=black] (17.6, 1.6) -- (19.2, 0.0);
\node [color=black] at (17.952, 0.8) {$\scriptstyle I_{v_2}$};
\node [color=black] at (18.848, 0.8) {$\scriptstyle O_{v_3}$};
\node [color=black] at (18.4, 0.288) {$\scriptstyle X$};
\node [color=black] at (18.4, 1.312) {$\scriptstyle X$};
\end{tikzpicture}

%% file: exampletiling-algo-eliminated_paper.tikz
\begin{tikzpicture}
\draw (0.0,0.0) rectangle (16.0,1.6);
\node (blank) at (0.8, 0.8) [draw,fill=white,thick,minimum width=1.6cm, minimum height=1.6cm] {};
\node (blank) at (2.4, 0.8) [draw,fill=white,thick,minimum width=1.6cm, minimum height=1.6cm] {};
\node (blank) at (4.0, 0.8) [draw,fill=white,thick,minimum width=1.6cm, minimum height=1.6cm] {};
\node [anchor=north,fill=white] at (5.6, 0.0) {\large $\mathcal{T}_{v_1}$};
\node (V1) at (5.6, 0.8) [rectangle,draw,fill=gray!30,thick,minimum width=1.6cm, minimum height=1.6cm] {};
\draw [color=black] (4.8, 0.0) -- (6.4, 1.6);
\draw [color=black] (4.8, 1.6) -- (6.4, 0.0);
\node [color=black] at (5.152, 0.8) {$\scriptstyle I_{v_1}$};
\node [color=black] at (6.048, 0.8) {$\scriptstyle O_{v_1}$};
\node [color=black] at (5.6, 0.288) {$\scriptstyle U_{v_1}$};
\node [color=black] at (5.6, 1.312) {$\scriptstyle U_{v_1}$};
\node [anchor=north,fill=white] at (7.2, 0.0) {\large $\mathcal{T}_{e_{14}}$};
\node (E14) at (7.2, 0.8) [rectangle,draw,fill=white,thick,minimum width=1.6cm, minimum height=1.6cm] {};
\draw [color=black] (6.4, 0.0) -- (8.0, 1.6);
\draw [color=black] (6.4, 1.6) -- (8.0, 0.0);
\node [color=black] at (6.752, 0.8) {$\scriptstyle O_{v_1}$};
\node [color=black] at (7.648, 0.8) {$\scriptstyle I_{v_4}$};
\node [color=black] at (7.2, 0.288) {$\scriptstyle X$};
\node [color=black] at (7.2, 1.312) {$\scriptstyle X$};
\node [anchor=north,fill=white] at (8.8, 0.0) {\large $\mathcal{T}_{v_4}$};
\node (V4) at (8.8, 0.8) [rectangle,draw,fill=gray!30,thick,minimum width=1.6cm, minimum height=1.6cm] {};
\draw [color=black] (8.0, 0.0) -- (9.6, 1.6);
\draw [color=black] (8.0, 1.6) -- (9.6, 0.0);
\node [color=black] at (8.352, 0.8) {$\scriptstyle I_{v_4}$};
\node [color=black] at (9.248, 0.8) {$\scriptstyle O_{v_4}$};
\node [color=black] at (8.8, 0.288) {$\scriptstyle U_{v_4}$};
\node [color=black] at (8.8, 1.312) {$\scriptstyle U_{v_4}$};
\node [anchor=north,fill=white] at (10.4, 0.0) {\large $\mathcal{T}_{e_{43}}$};
\node (E43) at (10.4, 0.8) [rectangle,draw,fill=white,thick,minimum width=1.6cm, minimum height=1.6cm] {};
\draw [color=black] (9.6, 0.0) -- (11.2, 1.6);
\draw [color=black] (9.6, 1.6) -- (11.2, 0.0);
\node [color=black] at (9.952, 0.8) {$\scriptstyle O_{v_4}$};
\node [color=black] at (10.848, 0.8) {$\scriptstyle I_{v_3}$};
\node [color=black] at (10.4, 0.288) {$\scriptstyle X$};
\node [color=black] at (10.4, 1.312) {$\scriptstyle X$};
\node [anchor=north,fill=white] at (12.0, 0.0) {\large $\mathcal{T}_{v_3}$};
\node (V3) at (12.0, 0.8) [rectangle,draw,fill=gray!30,thick,minimum width=1.6cm, minimum height=1.6cm] {};
\draw [color=black] (11.2, 0.0) -- (12.8, 1.6);
\draw [color=black] (11.2, 1.6) -- (12.8, 0.0);
\node [color=black] at (11.552, 0.8) {$\scriptstyle I_{v_3}$};
\node [color=black] at (12.448, 0.8) {$\scriptstyle O_{v_3}$};
\node [color=black] at (12.0, 0.288) {$\scriptstyle U_{v_3}$};
\node [color=black] at (12.0, 1.312) {$\scriptstyle U_{v_3}$};
\node (blank) at (13.6, 0.8) [draw,fill=white,thick,minimum width=1.6cm, minimum height=1.6cm] {};
\node (blank) at (15.2, 0.8) [draw,fill=white,thick,minimum width=1.6cm, minimum height=1.6cm] {};
\node [anchor=north,fill=white] at (16.8, 0.0) {\large $\mathcal{T}_{v_2}$};
\node (V2) at (16.8, 0.8) [rectangle,draw,fill=gray!30,thick,minimum width=1.6cm, minimum height=1.6cm] {};
\draw [color=black] (16.0, 0.0) -- (17.6, 1.6);
\draw [color=black] (16.0, 1.6) -- (17.6, 0.0);
\node [color=black] at (16.352, 0.8) {$\scriptstyle O_{v_2}$};
\node [color=black] at (17.248, 0.8) {$\scriptstyle I_{v_2}$};
\node [color=black] at (16.8, 0.288) {$\scriptstyle U_{v_2}$};
\node [color=black] at (16.8, 1.312) {$\scriptstyle U_{v_2}$};
\node (blank) at (18.4, 0.8) [draw,fill=white,thick,minimum width=1.6cm, minimum height=1.6cm] {};
\end{tikzpicture}

%% file: path-cover.tikz
\begin{tikzpicture}

\matrix[nodes={draw}, row sep=5mm, column sep=5mm] {
 & \node[circle] (V1) {$v_1$}; & \node[circle] (V4) {$v_4$}; & \node[circle] (V3) {$v_3$}; & \node[circle] (V2) {$v_2$}; & \node[circle] (V5) {$v_5$};  \\
};

\path (V1) edge[-latex,thick] (V4);
\path (V4) edge[-latex,thick] (V3);
\end{tikzpicture}

%% file: gadgets.tikz
\begin{tikzpicture}[scale=0.5]

\node [circle, draw, inner sep=0pt, minimum width=10pt] (v4) at (-6,6) {};
\node [circle, draw, inner sep=0pt, minimum width=10pt] (v3) at (-7,5) {};
\node [circle, draw, inner sep=0pt, minimum width=10pt] (v2) at (-7,3) {};
\node [circle, draw, inner sep=0pt, minimum width=10pt] (v1) at (-6,2) {};
\node [circle, draw, inner sep=0pt, minimum width=10pt] (v8) at (-7,1) {};
\node [circle, draw, inner sep=0pt, minimum width=10pt] (v7) at (-7,-1) {};
\node [circle, draw, inner sep=0pt, minimum width=10pt] (v9) at (-6,-2) {};
\node [circle, draw, inner sep=0pt, minimum width=10pt] (v6) at (-7,-3) {};
\node [circle, draw, inner sep=0pt, minimum width=10pt] (v11) at (-7,-5) {};
\node [circle, draw, inner sep=0pt, minimum width=10pt] (v10) at (-6,-6) {};
\node [circle, draw, inner sep=0pt, minimum width=10pt] (v13) at (-8,-2) {};
\node [circle, draw, inner sep=0pt, minimum width=10pt] (v12) at (-8,2) {};
\node [circle, draw, inner sep=0pt, minimum width=10pt] (v5) at (-9,0) {};
\node [circle, draw=none, fill=gray!20, inner sep=0pt, minimum width=15pt] at (-5,0) {};
\draw [-latex] (v9) edge (v7);
\draw [-latex] (v8) edge (v1);
\draw [-latex] (v2) edge (v12);
\draw [-latex] (v12) edge (v13);
\draw [-latex] (v10) edge (v11);
\draw [-latex] (v6) edge (v9);
\draw [-latex] (v7) edge (v13);
\draw [-latex] (v12) edge (v8);
\draw [-latex] (v5) edge (v12);
\draw [-latex] (v13) edge (v5);
\draw [-latex] (v1) edge (v2);
\draw [-latex] (v3) edge (v4);
\draw [-latex] (v13) edge (v6);
\draw [-latex, bend left] (v11) edge (v6);
\draw [-latex, bend left] (v6) edge (v11);
\draw [-latex, bend left] (v11) edge (v5);
\draw [-latex, bend left] (v5) edge (v3);
\draw [-latex, bend left] (v3) edge (v2);
\draw [-latex, bend left] (v2) edge (v3);
\draw [-latex, bend left] (v7) edge (v8);
\draw [-latex, bend left] (v8) edge (v7);
\draw [-latex, bend right] (v10) edge (v9);
\draw [-latex, bend right] (v9) edge (v1);
\draw [-latex, bend right] (v1) edge (v4);
\node [draw=none, inner sep=0pt] (v17) at (-6,8) {};
\node [draw=none, inner sep=0pt] (v18) at (-6,-8) {};
\draw [-, line width = 15, color=gray!20] (-5,0) edge (5,1.5);

\node [circle, draw, fill=black, inner sep=0pt, minimum width=3pt] (v23) at (-5,4) {};
\node [circle, draw, fill=black, inner sep=0pt, minimum width=3pt] (v21) at (-5,0) {};
\node [circle, draw, fill=black, inner sep=0pt, minimum width=3pt] (v33) at (-5,-4) {};

\draw [rounded corners = 10pt, dashed] (-10,7) -- (-4,7) -- (-4,-7) -- (-10,-7) -- cycle;
\draw [draw=none, rounded corners = 10pt, fill=gray!20] (4,5) -- (8,5) -- (8,-5) -- (4,-5) -- cycle;
\node [circle, draw, inner sep=0pt, minimum width=10pt] (v14) at (6,4) {};
\node [circle, draw, inner sep=0pt, minimum width=10pt] (v15) at (6,0) {};
\node [circle, draw, inner sep=0pt, minimum width=10pt] (v16) at (6,-4) {};
\draw [-latex, bend left] (v14) edge (v15);
\draw [-latex, bend left] (v15) edge (v16);
\draw [-latex, bend right] (v14) edge (v15);
\draw [-latex, bend right] (v15) edge (v16);
\node (v19) at (6,6) {};
\node (v20) at (6,-6) {};
\draw [-latex] (v4) edge (v17);
\draw [-latex] (v18) edge (v10);
\draw [-latex] (v19) edge (v14);
\draw [-latex] (v16) edge (v20);
\node [circle, draw, fill=black, inner sep=0pt, minimum width=3pt] (v22) at (5,1.5) {};
\node [circle, draw, fill=black, inner sep=0pt, minimum width=3pt] (v30) at (5,-2) {};
\node [circle, draw, fill=black, inner sep=0pt, minimum width=3pt] (v35) at (7,2) {};
\node [circle, draw, fill=black, inner sep=0pt, minimum width=3pt] (v36) at (7,-2) {};
\node [circle, draw, fill=black, inner sep=0pt, minimum width=3pt] (v32) at (5,-3) {};
\node [circle, draw, fill=black, inner sep=0pt, minimum width=3pt] (v28) at (5,-1) {};
\node [circle, draw, fill=black, inner sep=0pt, minimum width=3pt] (v26) at (5,2.5) {};

\node [draw=none,fill=none] (v24) at (-1,2) {};
\node (v34) at (-1,-5) {};
\node[rectangle, draw=none, shading = axis, rotate=-20.56, left color=gray!40, right color=white, shading angle=270-20.56,inner sep=0pt, minimum width=30pt, minimum height = 15pt] (v25) at (1, 4) {};
\node[rectangle, draw=none, shading = axis, rotate=-7.20, left color=gray!40, right color=white, shading angle=270,inner sep=0pt, minimum width=30pt, minimum height = 15pt] (v27) at (1, -0.5) {};
\node[rectangle, draw=none, shading = axis, rotate=14.04, left color=gray!40, right color=white, shading angle=270+14.04,inner sep=0pt, minimum width=30pt, minimum height = 15pt] (v29) at (1, -3) {};
\node[rectangle, draw=none, shading = axis, rotate=26.57, left color=gray!40, right color=white, shading angle=270+26.57,inner sep=0pt, minimum width=30pt, minimum height = 15pt] (v31) at (1, -5) {};
\draw [-, line width = 15, color=gray!20]  (1, 4) edge (v26);
\draw [-, line width = 15, color=gray!20]  (1, -0.5) edge (v28);
\draw [-, line width = 15, color=gray!20]  (1, -3) edge (v30);
\draw [-, line width = 15, color=gray!20]  (1, -5) edge (v32);
\draw [-] (v21) edge (v22);
\draw [-] (v23) edge (v24);
\draw [-] (1, 4) edge (v26);
\draw [-] (1, -0.5) edge (v28);
\draw [-] (1, -3) edge (v30);
\draw [-] (1, -5) edge (v32);
\draw [-] (v33) edge (v34);

\draw [bend left] (v35) edge (v36);
\draw [rounded corners = 10pt, dashed] (4,5) -- (8,5) -- (8,-5) -- (4,-5) -- cycle;

\node at (-3.5,7.5) {$c_j$};
\node at (8.5,5.5) {$x_i$};
\node at (5.75,2) {$x_i$};
\node at (5.75,-2) {$\overline{x_i}$};
\end{tikzpicture}

%% file: connect.tikz
\begin{tikzpicture} [scale=0.5]

\node[circle, draw, radius=20pt] (v10) at (0,0) {$t$};
\node[circle, draw, radius=20pt] (v1) at (6,0) {$s$};
\node[rounded corners=2.5pt, draw, dashed, minimum height=20pt, minimum width=20pt] (v9) at (0,-3) {$c_1$};
\node[rounded corners=2.5pt, draw, dashed, minimum height=20pt, minimum width=20pt] (v8) at (0,-6) {$c_2$};
\node[rounded corners=2.5pt, dashed, minimum height=20pt, minimum width=20pt,text height=10pt] (v7) at (0,-9) {$\vdots$};
\node[rounded corners=2.5pt, draw, dashed, minimum height=20pt, minimum width=20pt] (v6) at (0,-12) {$c_m$};

\node[rounded corners=2.5pt, draw, dashed, minimum height=20pt, minimum width=20pt] (v2) at (6,-3) {$x_1$};
\node[rounded corners=2.5pt, draw, dashed, minimum height=20pt, minimum width=20pt] (v3) at (6,-6) {$x_2$};
\node[rounded corners=2.5pt, dashed, minimum height=20pt, minimum width=20pt,text height=10pt] (v4) at (6,-9) {$\vdots$};
\node[rounded corners=2.5pt, draw, dashed, minimum height=20pt, minimum width=20pt] (v5) at (6,-12) {$x_n$};

\node[rounded corners=2.5pt, dashed, minimum height=20pt, minimum width=20pt,text height=10pt] at (3,-9) {$\vdots$};

\node[circle, draw, fill=black, inner sep=0pt, minimum width=2.5pt] (v11) at (1,-2.5) {};
\node[circle, draw, fill=black, inner sep=0pt, minimum width=2.5pt] (v13) at (1,-3) {};
\node[circle, draw, fill=black, inner sep=0pt, minimum width=2.5pt] (v15) at (1,-3.5) {};

\node[circle, draw, fill=black, inner sep=0pt, minimum width=2.5pt] (v17) at (1,-5.75) {};
\node[circle, draw, fill=black, inner sep=0pt, minimum width=2.5pt] (v19) at (1,-6.25) {};

\node[circle, draw, fill=black, inner sep=0pt, minimum width=2.5pt] (v29) at (1,-11.5) {};
\node[circle, draw, fill=black, inner sep=0pt, minimum width=2.5pt] (v30) at (1,-12) {};
\node[circle, draw, fill=black, inner sep=0pt, minimum width=2.5pt] (v31) at (1,-12.5) {};

\draw [-latex] (v1) edge (v2);
\draw [-latex] (v2) edge (v3);
\draw [-latex] (v3) edge (v4);
\draw [-latex] (v4) edge (v5);
\draw [-latex] (v6) edge (v7);
\draw [-latex] (v7) edge (v8);
\draw [-latex] (v8) edge (v9);
\draw [-latex] (v9) edge (v10);

\draw [-latex, bend left = 60] (v5) edge (v6);

\node[circle, draw, fill=black, inner sep=0pt, minimum width=2.5pt] (v12) at (5,-2.5) {};
\node[circle, draw, fill=black, inner sep=0pt, minimum width=2.5pt] (v18) at (5,-2.75) {};
\node[circle, draw, fill=black, inner sep=0pt, minimum width=2.5pt] (v21) at (5,-3.25) {};
\node[circle, draw, fill=black, inner sep=0pt, minimum width=2.5pt] (v23) at (5,-3.5) {};

\node[circle, draw, fill=black, inner sep=0pt, minimum width=2.5pt] (v20) at (5,-5.5) {};
\node[circle, draw, fill=black, inner sep=0pt, minimum width=2.5pt] (v14) at (5,-6) {};
\node[circle, draw, fill=black, inner sep=0pt, minimum width=2.5pt] (v27) at (5,-6.25) {};
\node[circle, draw, fill=black, inner sep=0pt, minimum width=2.5pt] (v25) at (5,-6.5) {};

\node[circle, draw, fill=black, inner sep=0pt, minimum width=2.5pt] (v32) at (5,-11.625) {};
\node[circle, draw, fill=black, inner sep=0pt, minimum width=2.5pt] (v33) at (5,-11.875) {};
\node[circle, draw, fill=black, inner sep=0pt, minimum width=2.5pt] (v34) at (5,-12.375) {};

\node [circle, inner sep=0pt] (v22) at (1.25,-9.5) {};
\node [circle, inner sep=0pt] (v24) at (2,-9.5) {};
\node [circle, inner sep=0pt] (v35) at (2.5,-9.5) {};
\node [circle, inner sep=0pt] (v16) at (4,-9.5) {};
\node [circle, inner sep=0pt] (v26) at (4.55,-9.5) {};
\node [circle, inner sep=0pt] (v28) at (3.75,-9.5) {};
\node [circle, inner sep=0pt] (w22) at (1.25,-8.5) {};
\node [circle, inner sep=0pt] (w24) at (2,-8.5) {};
\node [circle, inner sep=0pt] (w16) at (3,-8.5) {};
\node [circle, inner sep=0pt] (w26) at (4.25,-8.5) {};
\node [circle, inner sep=0pt] (w28) at (3.5,-8.5) {};
\draw [-] (v11) edge (v12);
\draw [-] (v13) edge (v14);
\draw [-] (v15) edge (w26);
\draw [-] (v17) edge (v18);
\draw [-] (v19) edge (v20);
\draw [-] (v29) edge (v24);
\draw [-] (v30) edge (v28);
\draw [-] (v31) edge (v26);
\draw [-] (v32) edge (v16);
\draw [-] (v33) edge (v22);
\draw [-] (v34) edge (v35);
\draw [-] (v21) edge (w22);
\draw [-] (v23) edge (w16);
\draw [-] (v27) edge (w24);
\draw [-] (v25) edge (w28);
\end{tikzpicture}

%% file: xor-short.tikz
\begin{tikzpicture}[scale=0.5]

\node (v1) at (-6,0) {};
\node (v3) at (-3,0) {};

\node (v2) at (-6,10) {};
\node (v4) at (-3,10) {};

\node [circle, draw, fill=black, inner sep=0pt, minimum width=3pt] (v41) at (-5.75,5) {};
\node [circle, draw, fill=black, inner sep=0pt, minimum width=3pt] (v42) at (-3.25,5) {};

\draw [-latex] (v1) edge (v2);
\draw [-latex] (v4) edge (v3);
\draw [-] (v41) edge (v42);
\end{tikzpicture}

%% file: xor-real.tikz
\begin{tikzpicture}[scale=0.5]

\node [draw, circle, radius = 10pt] (v7) at (0,8) {};
\node [draw, circle, radius = 10pt] (v14) at (0,6) {};
\node [draw, circle, radius = 10pt] (v13) at (0,4) {};
\node [draw, circle, radius = 10pt] (v6) at (0,2) {};
\node [draw, circle, radius = 10pt] (v10) at (3,8) {};
\node [draw, circle, radius = 10pt] (v15) at (3,6) {};
\node [draw, circle, radius = 10pt] (v12) at (3,4) {};
\node [draw, circle, radius = 10pt] (v11) at (3,2) {};

\node (v5) at (0,0) {};
\node (v16) at (3,0) {};

\node (v8) at (0,10) {};
\node (v9) at (3,10) {};

\draw [-latex] (v5) edge (v6);
\draw [-latex] (v6) edge (v11);
\draw [-latex] (v11) edge (v12);
\draw [-latex] (v12) edge (v13);
\draw [-latex] (v13) edge (v14);
\draw [-latex] (v14) edge (v15);
\draw [-latex] (v15) edge (v10);
\draw [-latex] (v10) edge (v7);
\draw [-latex] (v7) edge (v8);
\draw [-latex] (v9) edge (v10);
\draw [-latex] (v7) edge (v14);
\draw [-latex] (v15) edge (v12);
\draw [-latex] (v13) edge (v6);
\draw [-latex] (v11) edge (v16);
\end{tikzpicture}

%% file: xor-noendpts.tikz
\begin{tikzpicture}[scale=0.5]

\node [draw, circle, radius = 10pt] (v19) at (6,8) {};
\node [draw, circle, radius = 10pt] (v20) at (6,6) {};
\node [draw, circle, radius = 10pt] (v23) at (6,4) {};
\node [draw, circle, radius = 10pt] (v24) at (6,2) {};
\node [draw, circle, radius = 10pt] (v18) at (9,8) {};
\node [draw, circle, radius = 10pt] (v21) at (9,6) {};
\node [draw, circle, radius = 10pt] (v22) at (9,4) {};
\node [draw, circle, radius = 10pt] (v25) at (9,2) {};
\node [draw, circle, radius = 10pt] (v35) at (11,8) {};
\node [draw, circle, radius = 10pt] (v32) at (11,6) {};
\node [draw, circle, radius = 10pt] (v31) at (11,4) {};
\node [draw, circle, radius = 10pt] (v28) at (11,2) {};
\node [draw, circle, radius = 10pt] (v34) at (14,8) {};
\node [draw, circle, radius = 10pt] (v33) at (14,6) {};
\node [draw, circle, radius = 10pt] (v30) at (14,4) {};
\node [draw, circle, radius = 10pt] (v29) at (14,2) {};

\node (v37) at (6,0) {};
\node (v26) at (9,0) {};
\node (v27) at (11,0) {};
\node (v40) at (14,0) {};

\node (v38) at (6,10) {};
\node (v17) at (9,10) {};
\node (v36) at (11,10) {};
\node (v39) at (14,10) {};

\draw [-latex] (v17) edge (v18);
\draw [-latex] (v19) edge (v20);
\draw [-latex] (v21) edge (v22);
\draw [-latex] (v23) edge (v24);
\draw [-latex] (v25) edge (v26);
\draw [-latex] (v27) edge (v28);
\draw [-latex] (v29) edge (v30);
\draw [-latex] (v31) edge (v32);
\draw [-latex] (v33) edge (v34);
\draw [-latex] (v35) edge (v36);
\draw [-latex, ultra thick] (v37) edge (v24);
\draw [-latex, ultra thick] (v24) edge (v25);
\draw [-latex, ultra thick] (v25) edge (v22);
\draw [-latex, ultra thick] (v22) edge (v23);
\draw [-latex, ultra thick] (v23) edge (v20);
\draw [-latex, ultra thick] (v20) edge (v21);
\draw [-latex, ultra thick] (v21) edge (v18);
\draw [-latex, ultra thick] (v18) edge (v19);
\draw [-latex, ultra thick] (v19) edge (v38);
\draw [-latex, ultra thick] (v39) edge (v34);
\draw [-latex, ultra thick] (v34) edge (v35);
\draw [-latex, ultra thick] (v35) edge (v32);
\draw [-latex, ultra thick] (v32) edge (v33);
\draw [-latex, ultra thick] (v33) edge (v30);
\draw [-latex, ultra thick] (v30) edge (v31);
\draw [-latex, ultra thick] (v31) edge (v28);
\draw [-latex, ultra thick] (v28) edge (v29);
\draw [-latex, ultra thick] (v29) edge (v40);
\end{tikzpicture}

%% file: assignment.tikz
\begin{tikzpicture}[scale=0.5]

\node [circle, draw, inner sep=0pt, minimum width=10pt] (v4) at (-6,-2) {};
\node [circle, draw, inner sep=0pt, minimum width=10pt] (v3) at (-7,-3) {};
\node [circle, draw, inner sep=0pt, minimum width=10pt] (v2) at (-7,-5) {};
\node [circle, draw, inner sep=0pt, minimum width=10pt] (v1) at (-6,-6) {};
\node [circle, draw, inner sep=0pt, minimum width=10pt] (v9) at (-6,6) {};
\node [circle, draw, inner sep=0pt, minimum width=10pt] (v6) at (-7,5) {};
\node [circle, draw, inner sep=0pt, minimum width=10pt] (v11) at (-7,3) {};
\node [circle, draw, inner sep=0pt, minimum width=10pt] (v10) at (-6,2) {};
\node [circle, draw, inner sep=0pt, minimum width=10pt] (v13) at (-8,6) {};
\node [circle, draw, inner sep=0pt, minimum width=10pt] (v12) at (-8,-6) {};

\node [circle, draw, fill=black, inner sep=0pt, minimum width=3pt] (v23) at (-5,-4) {};
\node [circle, draw, fill=black, inner sep=0pt, minimum width=3pt] (v37) at (-5,4) {};

\draw [rounded corners = 10pt, dashed] (-10,8.5) -- (-4,8.5) -- (-4,1) -- (-10,1) -- cycle;
\draw [rounded corners = 10pt, dashed] (-10,-1) -- (-4,-1) -- (-4,-8.5) -- (-10,-8.5) -- cycle;
\draw [rounded corners = 10pt, dashed] (0,5) -- (4,5) -- (4,-5) -- (0,-5) -- cycle;
\node [circle, draw, inner sep=0pt, minimum width=10pt] (v14) at (2,4) {};
\node [circle, draw, inner sep=0pt, minimum width=10pt] (v15) at (2,0) {};
\node [circle, draw, inner sep=0pt, minimum width=10pt] (v16) at (2,-4) {};
\node (v19) at (2,6) {};
\node (v20) at (2,-6) {};
\node [circle, draw, fill=black, inner sep=0pt, minimum width=3pt] (v22) at (1,2) {};
\node [circle, draw, fill=black, inner sep=0pt, minimum width=3pt] (v30) at (1,-2) {};
\node [circle, draw, fill=black, inner sep=0pt, minimum width=3pt] (v35) at (3,2) {};
\node [circle, draw, fill=black, inner sep=0pt, minimum width=3pt] (v36) at (3,-2) {};

\draw [-latex, bend right] (v10) edge (v9);
\draw [-latex, bend right] (v15) edge (v16);
\draw [-latex, bend left] (v14) edge (v15);
\draw [-latex, bend left] (v6) edge (v11);
\draw [-latex, bend left] (v11) edge (v6);
\draw [-latex, bend left] (v3) edge (v2);
\draw [-latex, bend left] (v2) edge (v3);
\node (v5) at (-9,-7) {};
\node (v8) at (-7,-7) {};
\node (v7) at (-8,-7) {};
\node (v17) at (-5.25,-7) {};
\draw [-latex, bend left] (v5) edge (v3);
\draw [-latex] (v5) edge (v12);
\draw [-latex] (v12) edge (v7);
\draw [-latex] (v12) edge (v8);
\draw [-latex] (v3) edge (v4);
\draw [-latex] (v2) edge (v12);
\draw [-latex] (v1) edge (v2);
\draw [-latex] (v8) edge (v1);
\draw [-latex] (v10) edge (v11);
\draw [-latex] (v6) edge (v9);
\draw [-latex] (v13) edge (v6);
\draw [-latex,bend right] (v17) edge (v1);
\node (v31) at (-9,7) {};
\node (v32) at (-8,7) {};
\node (v33) at (-7,7) {};
\node (v34) at (-5.25,7) {};
\draw [-latex,bend right] (v9) edge (v34);
\draw [-latex,bend left] (v13) edge (v31);
\draw [-latex, bend left] (v11) edge (v31);
\draw [-latex] (v32) edge (v13);
\draw [-latex] (v33) edge (v13);
\draw [-latex] (v9) edge (v33);
\node [text height=10pt,minimum height=20pt] at (-7,7.5) {$\vdots$};
\node [text height=10pt,minimum height=20pt] at (-7,-7.5) {$\vdots$};
\node [text height=10pt,,minimum height=20pt] at (2,6.5) {$\vdots$};
\node [text height=10pt,minimum height=20pt] at (2,-6.5) {$\vdots$};
\draw  (v22) edge (v37);
\draw  (v30) edge (v23);
\node at (1.75,2) {${x_i}$};
\node at (1.75,-2) {$\overline{x_i}$};

\draw [-latex, ultra thick] (v4) edge (v10);
\draw [-latex, ultra thick, bend right] (v1) edge (v4);
\draw [-latex, ultra thick, bend right] (v14) edge (v15);
\draw [-latex, ultra thick, bend left] (v15) edge (v16);
\draw [-latex, ultra thick] (v19) edge (v14);
\draw [-latex, ultra thick] (v16) edge (v20);
\draw [-, bend left] (v35) edge (v36);
\node at (-5,3) {$e_1$};
\node at (-5,-5) {$e_2$};
\end{tikzpicture}

%% file: satisfy.tikz
\begin{tikzpicture}[scale=0.5]

\node [circle, draw, inner sep=0pt, minimum width=10pt] (a4) at (-7,6) {};
\node [circle, draw, inner sep=0pt, minimum width=10pt] (a3) at (-8,5) {};
\node [circle, draw, inner sep=0pt, minimum width=10pt] (a2) at (-8,3) {};
\node [circle, draw, inner sep=0pt, minimum width=10pt] (a1) at (-7,2) {};
\node [circle, draw, inner sep=0pt, minimum width=10pt] (a8) at (-8,1) {};
\node [circle, draw, inner sep=0pt, minimum width=10pt] (a7) at (-8,-1) {};
\node [circle, draw, inner sep=0pt, minimum width=10pt] (a9) at (-7,-2) {};
\node [circle, draw, inner sep=0pt, minimum width=10pt] (a6) at (-8,-3) {};
\node [circle, draw, inner sep=0pt, minimum width=10pt] (a11) at (-8,-5) {};
\node [circle, draw, inner sep=0pt, minimum width=10pt] (a10) at (-7,-6) {};
\node [circle, draw, inner sep=0pt, minimum width=10pt] (a13) at (-9,-2) {};
\node [circle, draw, inner sep=0pt, minimum width=10pt] (a12) at (-9,2) {};
\node [circle, draw, inner sep=0pt, minimum width=10pt] (a5) at (-10,0) {};
\node [draw=none, inner sep=0pt] (a17) at (-7,8) {};
\node [draw=none, inner sep=0pt] (a18) at (-7,-8) {};

\draw [-latex] (a9) edge (a7);
\draw [-latex] (a8) edge (a1);
\draw [-latex] (a2) edge (a12);
\draw [-latex] (a12) edge (a13);
\draw [-latex] (a10) edge (a11);
\draw [-latex] (a6) edge (a9);
\draw [-latex] (a7) edge (a13);
\draw [-latex] (a12) edge (a8);
\draw [-latex] (a5) edge (a12);
\draw [-latex] (a13) edge (a5);
\draw [-latex] (a1) edge (a2);
\draw [-latex] (a3) edge (a4);
\draw [-latex] (a13) edge (a6);
\draw [-latex, bend left] (a11) edge (a6);
\draw [-latex, bend left] (a6) edge (a11);
\draw [-latex, bend left] (a11) edge (a5);
\draw [-latex, bend left] (a5) edge (a3);
\draw [-latex, bend left] (a3) edge (a2);
\draw [-latex, bend left] (a2) edge (a3);
\draw [-latex, bend left] (a7) edge (a8);
\draw [-latex, bend left] (a8) edge (a7);
\draw [-latex, bend right] (a10) edge (a9);
\draw [-latex, bend right] (a9) edge (a1);
\draw [-latex, bend right] (a1) edge (a4);
\draw [-latex] (a18) edge (a10);
\draw [-latex] (a4) edge (a17);

\node [circle, draw, inner sep=0pt, minimum width=10pt] (c4) at (5,6) {};
\node [circle, draw, inner sep=0pt, minimum width=10pt] (c3) at (4,5) {};
\node [circle, draw, inner sep=0pt, minimum width=10pt] (c2) at (4,3) {};
\node [circle, draw, inner sep=0pt, minimum width=10pt] (c1) at (5,2) {};
\node [circle, draw, inner sep=0pt, minimum width=10pt] (c8) at (4,1) {};
\node [circle, draw, inner sep=0pt, minimum width=10pt] (c7) at (4,-1) {};
\node [circle, draw, inner sep=0pt, minimum width=10pt] (c9) at (5,-2) {};
\node [circle, draw, inner sep=0pt, minimum width=10pt] (c6) at (4,-3) {};
\node [circle, draw, inner sep=0pt, minimum width=10pt] (c11) at (4,-5) {};
\node [circle, draw, inner sep=0pt, minimum width=10pt] (c10) at (5,-6) {};
\node [circle, draw, inner sep=0pt, minimum width=10pt] (c13) at (3,-2) {};
\node [circle, draw, inner sep=0pt, minimum width=10pt] (c12) at (3,2) {};
\node [circle, draw, inner sep=0pt, minimum width=10pt] (c5) at (2,0) {};
\node [draw=none, inner sep=0pt] (c17) at (5,8) {};
\node [draw=none, inner sep=0pt] (c18) at (5,-8) {};

\draw [-latex] (c9) edge (c7);
\draw [-latex] (c8) edge (c1);
\draw [-latex] (c2) edge (c12);
\draw [-latex] (c12) edge (c13);
\draw [-latex] (c10) edge (c11);
\draw [-latex] (c6) edge (c9);
\draw [-latex] (c7) edge (c13);
\draw [-latex] (c12) edge (c8);
\draw [-latex] (c5) edge (c12);
\draw [-latex] (c13) edge (c5);
\draw [-latex] (c1) edge (c2);
\draw [-latex] (c3) edge (c4);
\draw [-latex] (c13) edge (c6);
\draw [-latex, bend left] (c11) edge (c6);
\draw [-latex, bend left] (c6) edge (c11);
\draw [-latex, bend left] (c11) edge (c5);
\draw [-latex, bend left] (c5) edge (c3);
\draw [-latex, bend left] (c3) edge (c2);
\draw [-latex, bend left] (c2) edge (c3);
\draw [-latex, bend left] (c7) edge (c8);
\draw [-latex, bend left] (c8) edge (c7);
\draw [-latex, bend right] (c10) edge (c9);
\draw [-latex, bend right] (c9) edge (c1);
\draw [-latex, bend right] (c1) edge (c4);
\draw [-latex] (c18) edge (c10);
\draw [-latex] (c4) edge (c17);

\node [circle, draw, inner sep=0pt, minimum width=10pt] (b4) at (-1,6) {};
\node [circle, draw, inner sep=0pt, minimum width=10pt] (b3) at (-2,5) {};
\node [circle, draw, inner sep=0pt, minimum width=10pt] (b2) at (-2,3) {};
\node [circle, draw, inner sep=0pt, minimum width=10pt] (b1) at (-1,2) {};
\node [circle, draw, inner sep=0pt, minimum width=10pt] (b8) at (-2,1) {};
\node [circle, draw, inner sep=0pt, minimum width=10pt] (b7) at (-2,-1) {};
\node [circle, draw, inner sep=0pt, minimum width=10pt] (b9) at (-1,-2) {};
\node [circle, draw, inner sep=0pt, minimum width=10pt] (b6) at (-2,-3) {};
\node [circle, draw, inner sep=0pt, minimum width=10pt] (b11) at (-2,-5) {};
\node [circle, draw, inner sep=0pt, minimum width=10pt] (b10) at (-1,-6) {};
\node [circle, draw, inner sep=0pt, minimum width=10pt] (b13) at (-3,-2) {};
\node [circle, draw, inner sep=0pt, minimum width=10pt] (b12) at (-3,2) {};
\node [circle, draw, inner sep=0pt, minimum width=10pt] (b5) at (-4,0) {};
\node [draw=none, inner sep=0pt] (b17) at (-1,8) {};
\node [draw=none, inner sep=0pt] (b18) at (-1,-8) {};

\draw [-latex] (b9) edge (b7);
\draw [-latex] (b8) edge (b1);
\draw [-latex] (b2) edge (b12);
\draw [-latex] (b12) edge (b13);
\draw [-latex] (b10) edge (b11);
\draw [-latex] (b6) edge (b9);
\draw [-latex] (b7) edge (b13);
\draw [-latex] (b12) edge (b8);
\draw [-latex] (b5) edge (b12);
\draw [-latex] (b13) edge (b5);
\draw [-latex] (b1) edge (b2);
\draw [-latex] (b3) edge (b4);
\draw [-latex] (b13) edge (b6);
\draw [-latex, bend left] (b11) edge (b6);
\draw [-latex, bend left] (b6) edge (b11);
\draw [-latex, bend left] (b11) edge (b5);
\draw [-latex, bend left] (b5) edge (b3);
\draw [-latex, bend left] (b3) edge (b2);
\draw [-latex, bend left] (b2) edge (b3);
\draw [-latex, bend left] (b7) edge (b8);
\draw [-latex, bend left] (b8) edge (b7);
\draw [-latex, bend right] (b10) edge (b9);
\draw [-latex, bend right] (b9) edge (b1);
\draw [-latex, bend right] (b1) edge (b4);
\draw [-latex] (b18) edge (b10);
\draw [-latex] (b4) edge (b17);

\draw [-latex, ultra thick] (a18) edge (a10);
\draw [-latex, ultra thick] (a10) edge (a11);
\draw [-latex, ultra thick] (a5) edge (a12);
\draw [-latex, ultra thick] (a12) edge (a13);
\draw [-latex, ultra thick] (a13) edge (a6);
\draw [-latex, ultra thick] (a6) edge (a9);
\draw [-latex, ultra thick] (a9) edge (a7);
\draw [-latex, ultra thick] (a8) edge (a1);
\draw [-latex, ultra thick] (a1) edge (a2);
\draw [-latex, ultra thick] (a3) edge (a4);
\draw [-latex, ultra thick] (a4) edge (a17);
\draw [-latex, ultra thick] (b18) edge (b10);
\draw [-latex, ultra thick] (b10) edge (b11);
\draw [-latex, ultra thick] (b2) edge (b12);
\draw [-latex, ultra thick] (b12) edge (b13);
\draw [-latex, ultra thick] (b13) edge (b6);
\draw [-latex, ultra thick] (b6) edge (b9);
\draw [-latex, ultra thick] (b9) edge (b7);
\draw [-latex, ultra thick] (b8) edge (b1);
\draw [-latex, ultra thick] (c18) edge (c10);
\draw [-latex, ultra thick] (c10) edge (c11);
\draw [-latex, ultra thick] (c2) edge (c12);
\draw [-latex, ultra thick] (c12) edge (c8);
\draw [-latex, ultra thick] (c7) edge (c13);
\draw [-latex, ultra thick] (c13) edge (c6);
\draw [-latex, ultra thick] (c6) edge (c9);
\draw [-latex, ultra thick] (c4) edge (c17);
\draw [-latex, ultra thick] (b4) edge (b17);
\draw [-latex, ultra thick, bend right] (b1) edge (b4);
\draw [-latex, ultra thick, bend right] (c9) edge (c1);
\draw [-latex, ultra thick, bend right] (c1) edge (c4);
\draw [-latex, ultra thick, bend left] (a11) edge (a5);
\draw [-latex, ultra thick, bend left] (a7) edge (a8);
\draw [-latex, ultra thick, bend left] (a2) edge (a3);
\draw [-latex, ultra thick, bend left] (b11) edge (b5);
\draw [-latex, ultra thick, bend left] (b5) edge (b3);
\draw [-latex, ultra thick, bend left] (b3) edge (b2);
\draw [-latex, ultra thick, bend left] (b7) edge (b8);
\draw [-latex, ultra thick, bend left] (c11) edge (c5);
\draw [-latex, ultra thick, bend left] (c5) edge (c3);
\draw [-latex, ultra thick, bend left] (c3) edge (c2);
\draw [-latex, ultra thick, bend left] (c8) edge (c7);
\end{tikzpicture}

%% file: notsatisfy.tikz
\begin{tikzpicture}[scale=0.5]

\draw [draw=none, rounded corners = 10pt, fill=gray!20] (-4.45,5.6) -- (-1.4,5.6) -- (-1.4,-5.6) -- (-4.45,-5.6) -- cycle;
\node [circle, draw, inner sep=0pt, minimum width=10pt] (v4) at (-1,6) {};
\node [circle, draw, inner sep=0pt, minimum width=10pt] (v3) at (-2,5) {};
\node [circle, draw, inner sep=0pt, minimum width=10pt] (v2) at (-2,3) {};
\node [circle, draw, inner sep=0pt, minimum width=10pt] (v1) at (-1,2) {};
\node [circle, draw, inner sep=0pt, minimum width=10pt] (v8) at (-2,1) {};
\node [circle, draw, inner sep=0pt, minimum width=10pt] (v7) at (-2,-1) {};
\node [circle, draw, inner sep=0pt, minimum width=10pt] (v9) at (-1,-2) {};
\node [circle, draw, inner sep=0pt, minimum width=10pt] (v6) at (-2,-3) {};
\node [circle, draw, inner sep=0pt, minimum width=10pt] (v11) at (-2,-5) {};
\node [circle, draw, inner sep=0pt, minimum width=10pt] (v10) at (-1,-6) {};
\node [circle, draw, inner sep=0pt, minimum width=10pt] (v13) at (-3,-2) {};
\node [circle, draw, inner sep=0pt, minimum width=10pt] (v12) at (-3,2) {};
\node [circle, draw, inner sep=0pt, minimum width=10pt] (v5) at (-4,0) {};

\draw [-latex] (v9) edge (v7);
\draw [-latex] (v8) edge (v1);
\draw [-latex] (v2) edge (v12);
\draw [-latex] (v12) edge (v13);
\draw [-latex] (v10) edge (v11);
\draw [-latex] (v6) edge (v9);
\draw [-latex] (v7) edge (v13);
\draw [-latex] (v12) edge (v8);
\draw [-latex] (v5) edge (v12);
\draw [-latex] (v13) edge (v5);
\draw [-latex] (v1) edge (v2);
\draw [-latex] (v3) edge (v4);
\draw [-latex] (v13) edge (v6);
\draw [-latex, bend left] (v11) edge (v6);
\draw [-latex, bend left] (v6) edge (v11);
\draw [-latex, bend left] (v11) edge (v5);
\draw [-latex, bend left] (v5) edge (v3);
\draw [-latex, bend left] (v3) edge (v2);
\draw [-latex, bend left] (v2) edge (v3);
\draw [-latex, bend left] (v7) edge (v8);
\draw [-latex, bend left] (v8) edge (v7);
\draw [-latex, bend right] (v10) edge (v9);
\draw [-latex, bend right] (v9) edge (v1);
\draw [-latex, bend right] (v1) edge (v4);
\node [draw=none, inner sep=0pt] (v17) at (-1,8) {};
\node [draw=none, inner sep=0pt] (v18) at (-1,-8) {};

\draw [-latex] (v18) edge (v10);
\draw [-latex] (v4) edge (v17);

\draw [-latex, ultra thick] (v18) edge (v10);
\draw [-latex, ultra thick] (v4) edge (v17);
\draw [-latex, ultra thick, bend right] (v10) edge (v9);
\draw [-latex, ultra thick, bend right] (v9) edge (v1);
\draw [-latex, ultra thick, bend right] (v1) edge (v4);

\end{tikzpicture}